\documentclass[aps,prd,reprint,showkeys,nofootinbib,superscriptaddress,notitlepage]{revtex4-1}
\usepackage[utf8]{inputenc}
\usepackage{color}
\usepackage{verbatim}
\usepackage{prettyref}
\usepackage{mathtools}
\usepackage{mathrsfs}
\usepackage{amsmath}
\usepackage{amsthm}
\usepackage{amsfonts}
\usepackage{amssymb}
\usepackage[cal=dutchcal]{mathalfa}
\usepackage{booktabs}
\usepackage{tabu}
\usepackage{multirow}
\usepackage{verbatim}
\usepackage{enumitem}
\usepackage[unicode=true,
  bookmarks=false,   backref=false,    colorlinks=true,
  linktocpage=true,  citecolor=black,  linkcolor=black,
  urlcolor=black,    breaklinks=true
]{hyperref}

\usepackage{wasysym}

\makeatletter

\theoremstyle{plain}
\newtheorem{prop}{\protect\propositionname}
\providecommand{\propositionname}{Proposition}

\theoremstyle{remark}
\newtheorem*{rem*}{\protect\remarkname}
\providecommand{\remarkname}{Remark}

\theoremstyle{plain}
\newtheorem*{lem*}{\protect\lemmaname}
\providecommand{\lemmaname}{Lemma}

\theoremstyle{definition}
\newtheorem*{defn*}{\protect\definitionname}
\providecommand{\definitionname}{Definition}

\theoremstyle{definition}
\newtheorem{defn}{\protect\definitionname}
\providecommand{\definitionname}{Definition}

\theoremstyle{plain}
\newtheorem*{cor*}{\protect\corollaryname}
\providecommand{\corollaryname}{Corollary}

\usepackage{amsfonts}

\usepackage[dvipsnames]{xcolor}

\definecolor{darkBlue}{rgb}{0.15,0.15,0.65}
\definecolor{lightRed}{rgb}{1,0.7,0.8}

\allowdisplaybreaks

\usepackage{eso-pic, rotating, graphicx}

\newcommand{\bSe}{\begin{subequations}}
\newcommand{\eSe}{\end{subequations}}
\newcommand{\mydet}[1]{\operatorname{det}\left(#1\right)}
\newcommand{\dalamb}{\operatorname{\nabla ^2}}

\newcommand{\arel}{$A$-related }
\newcommand{\heta}{\hat{\eta} }
\newcommand{\hxi}{\hat{\xi} }

\makeatother

\makeatletter
\renewcommand*\env@matrix[1][\arraystretch]{%
  \edef\arraystretch{#1}%
  \hskip -\arraycolsep
  \let\@ifnextchar\new@ifnextchar
  \array{*\c@MaxMatrixCols c}}
\makeatother

\global\long\def\tr{\mathsf{{\scriptscriptstyle T}}}

\pdfmapfile{+rsfso.map}
\DeclareSymbolFont{rsfso}{U}{rsfso}{m}{n}
\DeclareSymbolFontAlphabet{\mathscr}{rsfso}

\usepackage[normalem]{ulem}

\usepackage{centernot}

\usepackage{array}
\newcolumntype{P}[1]{>{\centering\arraybackslash}p{#1}}

\AtBeginDocument{
\heavyrulewidth=.08em
\lightrulewidth=.05em
\cmidrulewidth=.03em
\belowrulesep=.65ex
\belowbottomsep=0pt
\aboverulesep=.4ex
\abovetopsep=0pt
\cmidrulesep=\doublerulesep
\cmidrulekern=.5em
\defaultaddspace=.5em
}

\usepackage{diagbox}
\usepackage{colortbl}

\begin{document}

\renewcommand{\theenumi}{(\roman{enumi})}

\title{Spacetime symmetries and topology in bimetric relativity}

\author{Francesco \surname{Torsello}}
\email{francesco.torsello@fysik.su.se}

\author{Mikica \surname{Kocic}}
\email{mikica.kocic@fysik.su.se}

\author{Marcus \surname{H\"og\aa s}}
\email{marcus.hogas@fysik.su.se}

\author{Edvard \surname{M\"ortsell}}
\email{edvard@fysik.su.se}
\affiliation{Department of Physics \& The Oskar Klein Centre, \\
Stockholm University, AlbaNova University Centre, SE-106 91 Stockholm, Sweden}


\begin{abstract}
We explore spacetime symmetries and topologies of the two metric sectors in Hassan--Rosen bimetric theory. We show that, in vacuum, the two sectors can either share or have separate spacetime symmetries. If stress--energy tensors are present, a third case can arise, with different spacetime symmetries within the same sector. This raises the question of the best definition of spacetime symmetry in Hassan--Rosen bimetric theory. We emphasize the possibility of imposing ansatzes and looking for solutions having different Killing vector fields or different isometries in the two sectors, which has gained little attention so far. We also point out that the topology of spacetime imposes a constraint on possible metric combinations.
\end{abstract}
\maketitle

\global\long\def\diag{\operatorname{diag}}
\global\long\def\dd{\mathrm{d}}
\global\long\def\ee{\mathrm{e}}
\global\long\def\ii{\mathrm{i}}
\global\long\def\eH{\mathrm{{\scriptscriptstyle H}}}

\newcommand{\e}{\mathrm{e}}
\newcommand{\rhg}{r^{g}_\mathrm{\scriptstyle{H}}}
\newcommand{\rhf}{r^{f}_\mathrm{\scriptstyle{H}}}

\section{Introduction}
\label{sec:introduction}

Symmetries are fundamental in physics. The invariance of physical quantities under a given transformation allows physicists to simplify their models of Nature and to better understand them.

The first computed exact solution in general relativity (GR) is the Schwarzschild solution \cite{Schwarzschild:1916uq}, which was determined by assuming a static and spherically symmetric spacetime. Schwarzschild's approach has been widely and fruitfully used until our days. Indeed, the standard way for finding particular solutions of certain field equations is to impose some symmetries to the system, and deduce the appropriate ansatz for the dynamical fields.

In (pseudo-)Riemannian geometry, spacetime symmetries are characterized by 
Killing vector fields (KVFs), whose integral curves define paths along 
which the physical quantities can be invariant. When this happens, the standard terminology is that the physical quantity possesses a symmetry along that path (or that the path defines a ``collineation" for the physical quantity). In the case of a symmetry of the metric tensor, they coincide with its isometries. We refer the reader to \cite{Wald:1984rg} for a basic treatment of spacetime symmetries and to \cite{hall2004symmetries} for a more advanced approach.

In this paper we explore the relation between the KVFs of two metrics defined on the same differentiable manifold. The 
motivation for this study comes from the Hassan--Rosen (HR) bimetric theory 
(introduced 
in \autoref{sec:HR}), where this geometrical framework naturally 
appears. When two metrics are concerned, the question of whether an isometry 
of one metric is a spacetime symmetry for the whole system arises. We will see 
that, if the two metrics share their isometries, then the latter will be spacetime symmetries for all other tensor fields. In this case, it is 
straightforward to talk about a spacetime symmetry, 
since 
the whole system is invariant under the transformation. If this is not the case, 
however, the definition of a spacetime symmetry is ambiguous. Solutions having different KVFs and showing this ambiguity, in the context of the HR bimetric theory, are presented in \autoref{subsec:explicit}.

When finding particular solutions, the most common ansatzes presume the same isometries in both sectors (some exceptions can be found in \cite{PhysRevD.86.061502,Nersisyan:2015oha}). This study investigates when this is the most general ansatz, and if 
there are other possible ansatzes which can, in principle, lead to 
new solutions of the bimetric field equations. This turns out to be the case.

We also answer to the question if there are any constraints on the solutions of HR bimetric theory, based on the underlying topology of the spacetime. The answer turns out to be yes, as we will discuss in \autoref{subsec:HR}.

The paper is organized in two main sections. In \autoref{sec:HR}, we focus on HR bimetric theory, reviewed in \autoref{subsec:HR}. We start by discussing explicit examples having the described properties in \autoref{subsec:explicit}. These examples should be thought of as a selection of cases showing different possible configurations of spacetime symmetries in HR theory. In principle, they may or may not provide any other physical information. In \autoref{subsec:isometriesHR} we study spacetime symmetries and we fit the examples into a more general framework. We determine a conserved vector current and state two propositions helpful for understanding the structure of spacetime symmetries in the theory. We define the concepts of ``bimetric spacetime symmetry", ``sectoral spacetime symmetry" and ``narrow spacetime symmetry" in bimetric relativity, as symmetries shared by all tensors, by the tensors in one metric sector only and by some generic tensors, respectively.
In Appendix \ref{app:ansatz}, we provide some strategies for constructing bimetric ansatzes having different KVFs. In \autoref{sec:geometry}, we fit the results of \autoref{sec:HR} into a wider perspective. In particular, we summarize the main geometric concepts and results we have used in studying spacetime symmetries. We explore the topic in a generic setup, i.e., we do not assume any relation between the two metrics, other than they are defined on the same differentiable manifold. After that, we state our conclusions.

The reader specifically interested in HR bimetric theory can find all the related results in \autoref{sec:HR} and Appendix \ref{app:ansatz}. The paper being quite technical---especially \autoref{sec:geometry}---the reader interested mainly in the results and their discussion is invited to read the brief summaries at the end of \autoref{sec:HR} and the beginning of \autoref{sec:geometry}, and the conclusions.

\section{Isometries in the Hassan-Rosen bimetric theory}
\label{sec:HR}

Recently, there has been a noteworthy interest in constructing a classical nonlinear theory of interacting spin-2 fields. This interest culminated in 
the discovery of two such theories which are free from the pathological Boulware--Deser ghost \cite{Boulware:1973my}, thanks to the particular form of the interaction potential between the spin-2 tensor fields. These theories are the de Rham--Gabadadze--Tolley (dRGT) massive gravity \cite{deRham:2010ik,deRham:2010kj}, which was proven to be ghost-free in \cite{Hassan:2011hr}, and the Hassan--Rosen (HR) bimetric theory \cite{Hassan:2011zd,Hassan:2011ea}, whose unambiguous definition and spacetime interpretation are provided in \cite{Hassan:2017ugh}.

We focus on the HR bimetric theory, describing two interacting, dynamical and symmetric spin-2  fields, $g_{\mu\nu}$ and 
$f_{\mu\nu}$. 
It has seven degrees of freedom propagating  
five massive and two massless mode around proportional backgrounds \cite{Comelli:2011wq,Hassan:2012wr,Hassan:2012rq}. One possible interpretation 
of this theory is to consider both the symmetric spin-2 fields as metrics defined 
on the same differentiable manifold. Here, we follow this interpretation and, 
therefore, aim to understand the relations between the KVFs of the two metrics and, consequently, between their isometries.

\subsection{Review of the theory, and topological constraint}
\label{subsec:HR}

The action of HR bimetric theory is \cite{Hassan:2011zd},
\begin{align}
\label{eq:Action}
\mathscr{S} & =\int\dd^{4}x\sqrt{-g}\left[\left(\frac{1}{2}M_{g}^{2}\,R_{g}+\mathcal{L}_g\right)-m^{4}V(S)\right. \nonumber \\
			&\qquad\qquad\qquad\ \ \,  \left. +\mydet{S} \left(\frac{1}{2}M_{f}^{2}\,R_{f}+\mathcal{L}_f\right)\right]
\end{align}
where two Einstein--Hilbert actions constitute kinetic terms for both metrics, minimally coupled to two independent matter sources described by the lagrangian densities $\mathcal{L}_g$, $\mathcal{L}_f$ (see, e.g., \cite{deRham:2014naa} regarding ghost-free matter couplings). The ``bimetric potential" is defined as,
\begin{equation}
\label{eq:potential}
V(S)\coloneqq\sum_{n=0}^{4}\beta_{n}e_{n}\left(S\right).
\end{equation}
Here, $S$ is the square root matrix defined by,
\begin{equation}
\label{eq:squareroot}
S\coloneqq \left(g^{-1}f\right)^\frac{1}{2},
\end{equation}
or, in index notation, ${S^\mu}_\rho{S^\rho}_\nu\coloneqq g^{\mu\rho}f_{\rho\nu}$.
The $\beta _n$ parameters are generic real numbers and the $e_{n}(S)$ are the elementary symmetric polynomials of ${S^\mu}_\nu$,
\begin{align}
\label{eq:ESP}
e_{0}\left(S\right)&=1, \nonumber\\
e_{1}\left(S\right)&=\lambda_1+\lambda_2+\lambda_3+\lambda_4, \nonumber \\
e_{2}\left(S\right)&=\lambda_1\lambda_2+\lambda_1\lambda_3+\lambda_1\lambda_4+\lambda_2\lambda_3+\lambda_2\lambda_4+\lambda_3\lambda_4, \nonumber \\
e_{3}\left(S\right)&=\lambda_1\lambda_2\lambda_3+\lambda_1\lambda_2\lambda_4+\lambda_1\lambda_3\lambda_4+\lambda_2\lambda_3\lambda_4, \nonumber \\
e_{4}\left(S\right)&=\lambda_1\lambda_2\lambda_3\lambda_4,
\end{align}
where the $\lambda_n$ are the eigenvalues of $S{}^\mu{}_\nu$.

Varying \eqref{eq:Action} with respect to $g_{\mu\nu}$ and $f_{\mu\nu}$, we get the bimetric field equations,
\begin{subequations}
\label{eq:EoM}
\begin{align}
\label{eq:EoMg}G_{g}{}^\mu{}_\nu+\dfrac{m^4}{M_{g}^2}\,V_{g}{}^\mu{}_\nu & =\dfrac{1}{M_g^2}T_g{}^\mu{}_\nu, \\
\label{eq:EoMf}G_{f}{}^\mu{}_\nu+\dfrac{m^4}{M_{f}^2}\,V_{f}{}^\mu{}_\nu & =\dfrac{1}{M_f^2}T_f{}^\mu{}_\nu,
\end{align}
\end{subequations}
with $G_{g,f}{}^\mu{}_\nu$ being the Einstein tensors for $g_{\mu\nu}$ and $f_{\mu\nu}$, $T_{g,f}{}^\mu{}_\nu$ being the stress--energy tensors for the two independent matter sources and $V_{g,f}{}^\mu{}_\nu$, here referred to as the ``tensor potentials."
being equal to,
\begin{subequations}
\begin{align}
\label{eq:Vg}
V_{g}{}^\mu{}_\nu & =\sum_{n=0}^3\beta_{n} \sum_{k=0}^{n}(-1)^{n+k}e_{k}(S)\,S^{n-k},\\
\label{eq:Vf}
V_{f}{}^\mu{}_\nu & =\sum_{n=0}^3\beta_{4-n} \sum_{k=0}^{n}(-1)^{n+k}e_{k}(S^{-1})\,S^{-n+k}.
\end{align}
\end{subequations}
Combining \eqref{eq:EoM} with their own traces we obtain,
\begin{subequations}
\label{eq:EoM2}
\begin{align}
\label{eq:EoM2g}R_{g}{}^\mu{}_\nu 	& =\dfrac{1}{M_{g}^2}\left[m^4\left(\dfrac{1}{2}V_g\delta {}^\mu{}_\nu-V_{g}{}^\mu{}_\nu \right)+T_g{}^\mu{}_\nu-\dfrac{1}{2}T_g\delta{}^\mu{}_\nu\right], \\[4mm]
\label{eq:EoM2f}R_{f}{}^\mu{}_\nu 	& =\dfrac{1}{M_{f}^2}\left[m^4\left(\dfrac{1}{2}V_f\delta {}^\mu{}_\nu-V_{f}{}^\mu{}_\nu \right)+T_f{}^\mu{}_\nu-\dfrac{1}{2}T_f\delta{}^\mu{}_\nu\right],
\end{align}
\end{subequations}
where $T_g=g_{\alpha\beta}{T_{g}}^{\alpha\beta}$, $T_f=f_{\alpha\beta}{T_{f}}^{\alpha\beta}$, $V_g=g_{\alpha\beta}{V_{g}}^{\alpha\beta}$ and $V_f=f_{\alpha\beta}{V_{f}}^{\alpha\beta}$.

Noting that the stress--energy tensors $T_g{}^\mu{}_\nu$ and $T_f{}^\mu{}_\nu$ are divergence-less due to diffeomorphism invariance of the matter actions $\int \dd ^4x\sqrt{-g}\,\mathcal{L}_g$ and $\int \dd ^4x\sqrt{-g}\mydet{S}\,\mathcal{L}_f$, the Bianchi constraints follows from taking the divergence of \eqref{eq:EoM},
\begin{align}
\label{eq:Bianchi}
\nabla_\nu \,V_g{}^\nu{}_\mu=0, \quad \tilde{\nabla}_\nu \,V_f{}^\nu{}_\mu=0,
\end{align}
where $\nabla _\mu$ is the compatible covariant derivative of $g_{\mu\nu}$ and $\tilde{\nabla} _\mu$ is the compatible covariant derivative of $f_{\mu\nu}$.

The following algebraic relation holds between the tensor potentials and the bimetric potential \cite{Hassan:2014vja},
\begin{align}
\label{eq:potentials}
V_g{}^\mu{}_\nu+\mydet{S}\,V_f{}^\mu{}_\nu=V(S)\,\delta{}^\mu{}_\nu,
\end{align}
which implies,
\begin{align}
\label{eq:potentialsTrace}
V_g+\mydet{S}\,V_f=4\,V(S).
\end{align}                                                                                                                                                                                                                                     
The interaction action $-m^4\int \dd ^4x\sqrt{-g}\,V(S)$ is invariant under generic diffeomorphisms. This implies that the two Bianchi constraints are not independent,
\begin{align}
\label{eq:Bianchi2}
\nabla_\nu \,V_g{}^\nu{}_\mu=-\mydet{S}\tilde{\nabla}_\nu \,V_f{}^\nu{}_\mu.
\end{align}
This was proved in \cite{Damour:2002ws} for diffeomorphisms equal to the identity at the boundary of the integration domain, and can be generalized to generic diffeomorphisms, thanks to the algebraic identity \eqref{eq:potentials}.

\paragraph{Topological constraint.}
The importance of topology when considering particular solutions in HR bimetric theory was briefly pointed out in \cite{Kocic:2017wwf}. Here we extend the discussion, which is independent of any field equations and it is applicable to every theory involving more than one metric tensor on the same differentiable manifold.

The action \eqref{eq:Action} is defined on a differentiable manifold, on which we have a set of smooth charts covering it (an atlas). The interaction term in the action constrains the theory to have only one diffeomorphism invariance (see, e.g., \cite{Damour:2002ws,Schmidt-May:2015vnx}).\footnote{Without the bimetric potential, the action \eqref{eq:Action} reduces to two decoupled copies of GR. In such case, the gauge group $G$ of the theory would be the direct sum of two separate diffeomorphism groups, one acting on the $g$-sector only and the other on the $f$-sector only \cite{Boulanger:2000rq}. The bimetric potential reduces $G$ to its diagonal subgroup \cite{Damour:2002ws}.} Therefore, the chosen atlas must be shared by both metrics. Since an atlas uniquely determines the topology of the manifold \cite[p.~20]{lang1995differential}, the two metrics must be compatible with this topology. We stress that topology is not determined by any field equations, which, being differential, are always local \cite[Figure 31.5]{misner1973gravitation}, \cite{PhysRev.128.919}.

We will clarify this issue in the next subsection when discussing explicit examples.

\subsection{Explicit examples with different KVFs}
\label{subsec:explicit}

We start out by discussing some explicit examples displaying different KVFs or different isometries, or both. In \autoref{subsec:isometriesHR}, we see how they fit into the general results.

In the literature, solutions of this kind were presented in \cite{PhysRevD.86.061502,Nersisyan:2015oha}. In particular, \cite{PhysRevD.86.061502} considers a non-bidiagonal ansatz with a Friedmann--Lema\^itre--Robertson--Walker (FLRW) $g_{\mu\nu}$ (homogeneous and isotropic) and an inhomogeneous $f_{\mu\nu}$, with a homogeneous perfect fluid coupled to $g_{\mu\nu}$. In their analysis, the authors of \cite{Nersisyan:2015oha} discuss (i) bidiagonal cosmological solutions with an FLRW $g_{\mu\nu}$ (homogeneous and isotropic) and a Lema\^itre $f_{\mu\nu}$ (only isotropic) with an inhomogeneous perfect fluid coupled to $g_{\mu\nu}$, and (ii) an FLRW $g_{\mu\nu}$ and a Bianchi Type I $f_{\mu\nu}$, with an anisotropic fluid coupled to $g_{\mu\nu}$. The three cases concern metrics having different isometry groups. Note also that the solutions with the inhomogeneous and the anisotropic perfect fluids fit in the discussion in \autoref{subsec:isometriesHR} about the definition of a spacetime symmetry in HR bimetric theory.

We remark that the topological constraints described in the previous section should be taken into account also for these solutions.

One way of finding solutions in GR is to put some assumptions on a metric and then generate the corresponding stress--energy tensor defined through $T^\mu{}_\nu \coloneqq M_g^2\,G^\mu{}_\nu$. This is called the ``Synge's method" \cite[Chapter~VIII]{synge1960relativity}. It is worth to stress that one has to be very lucky to find an ``acceptable" or non-pathological stress--energy tensor with this method. A typical example is the generalized Vaidya case where one assumes a specific form of the mass function, then trying to interpret the resulting stress--energy tensor (an arbitrary mass function will not work in general, but some will be ``sensible''; see section 6 in \cite{Ibohal:2004kk}).
In the bimetric case, we can similarly put some assumptions on both $g_{\mu\nu}$ and $S^\mu{}_\nu$ (or alternatively $g_{\mu\nu}$ and $h_{\mu\nu}$ so that $S^\mu{}_\nu=g^{\mu\rho} h_{\rho\nu}$) and then generate $T_g^\mu{}_\nu$ and $T_f^\mu{}_\nu$ using the bimetric field equations \eqref{eq:EoM}.
This is how the axially symmetric solution was obtained in \cite{Kocic:2017hve}.  
Moreover, one can also require some of the tensors to vanish identically, restricting the ansatz; for instance, $V_g^\mu{}_\nu=V_f^\mu{}_\nu=0$ or one can require $T_g^\mu{}_\nu=0$ but keeping arbitrary $T_f^\mu{}_\nu\ne0$. 
Of course, in the end the solution depends on the judgment and taste of what is physical or ``sensible,'' though some consistency checks must always be satisfied, like conservation laws.
Now, in general, $g_{\mu\nu}$ can have some collineations where $S{}^\mu{}_\nu$ can spoil them by removing or introducing new collineations.  Hence, $g_{\mu\nu}$, $f_{\mu\nu}$, $S{}^\mu{}_\nu$, $T_g^{\mu\nu}$ and $T_f^{\mu\nu}$ can have different isometries.

We discuss three examples:
\begin{enumerate}
	\item a vacuum solution of the bimetric field equations,
	\item a solution with two independent matter sources coupled with $g_{\mu\nu}$ and $f_{\mu\nu}$ in a very peculiar way,
	\item Einstein vacuum solutions with a freely specifiable scalar function in one metric.
\end{enumerate}

We emphasize that these examples are meant to show different possible configurations of spacetime symmetries in HR theory, rather than being useful as descriptions of physical systems.

\subsubsection*{Example I: Bi-Einstein vacuum solutions}

The first example constitutes a family of non-bidiagonal solutions including black holes (BHs) which are allowed to have different Killing horizons. In the bidiagonal case, corresponding to the Type I configuration in \cite{Hassan:2017ugh}, this is forbidden by the proposition in \cite{Deffayet:2011rh} and its extension in \cite[Appendix C]{PhysRevD.96.064003}. The following solution being non-bidiagonal, is allowed to have different Killing horizons. The same family of solutions was found in \cite{Comelli:2011wq}, but here we stress the fact that some of them, although satisfying the bimetric field equations, must be excluded due to the topological constraint presented in \autoref{subsec:HR}. In addition, we discuss the solution in a new perspective, analysing the isometries of the two metric sectors.

We assume a spherically symmetric ansatz for the two metrics and work in the outgoing Eddington--Finkelstein chart $(v,r,\theta,\phi)$,
\begin{align}
g_{\mu\nu}&=
	 \begin{pmatrix}
	   	-G(r) & 1 & 0 & 0 \\
	   	1 & 0 & 0 & 0 \\
	   	0 & 0 & r^2 & 0 \\
	   	0 & 0 & 0 & r^2 \sin^2(\theta ) \\
	  \end{pmatrix}\!, \nonumber \\
f_{\mu\nu}&=
	 \begin{pmatrix}
	   	-\e ^{2q(r)}F(r) & \e ^{2q(r)} & 0 & 0 \\
	   	\e ^{2q(r)} & 0 & 0 & 0 \\
	    0 & 0 & R^2(r)r^2 & 0 \\
	   	0 & 0 & 0 & R^2(r)r^2 \sin^2 (\theta ) \\
	  \end{pmatrix}\!, \nonumber \\
S{}^\mu{}_\nu&=
	 \begin{pmatrix}
	   	\e ^{q(r)} & 0 & 0 & 0 \\
	   	\dfrac{1}{2}\e ^{q(r)}\left(G(r)-F(r)\right) & \e ^{q(r)} & 0 & 0 \\
	    0 & 0 & R(r) & 0 \\
	   	0 & 0 & 0 & R(r) \\
	  \end{pmatrix}\!,
\end{align}
where $S{}^\mu{}_\nu$ is the principal square root of $g^{-1}f$.

Assuming that $G\neq F$, from the field equations \eqref{eq:EoMg},
\begin{align}
\label{eq:examplevacuumg}
R(r)\equiv R_0&=\dfrac{1}{\beta _3}\left( -\beta _2\pm \sqrt{\beta _2 ^2-\beta _1\beta _3} \right), \nonumber \\
q(r)&= \log (R_0), \nonumber \\
G(r)&=1-\dfrac{\rhg}{r}+\dfrac{\Lambda_g (R_0;\{\beta_i\})}{3}r^2, \nonumber \\
\Lambda_g (R_0;\{\beta_i\})&=\dfrac{m^2}{\beta_3}\left( \beta_1 \beta _2-\beta_0\beta_3+2R_0\left( \beta_2^2-\beta_1\beta_3 \right) \right),
\end{align}
where $i$ runs from 0 to 3. This solutions also satisfies the Bianchi constraints \eqref{eq:Bianchi}. Substituting the expressions in \eqref{eq:examplevacuumg} in the field equations \eqref{eq:EoMf} results in,
\begin{align}
\label{eq:examplevacuumf}
F(r)&=1-\dfrac{\rhf}{r}+\dfrac{\Lambda_f (R_0;\{\beta_i\})}{3}r^2, \nonumber \\
\Lambda_f (R_0;\{\beta_i\})&=\pm\dfrac{m^2}{3\kappa\left( 2\beta_2+R_0\beta _3 \right)}\nonumber \\
						   &\quad \times\left[ 2\beta_1\left(\beta_3-\dfrac{\beta_2\beta_4}{\beta_3} \right)\right.\nonumber\\
						   &\quad \quad \left.+R_0\left(3\beta_2\beta_3+\beta_1\beta_4-\dfrac{4\beta^2_2\beta_4}{\beta_3} \right)\right].
\end{align}
The scalar invariants of the square root matrix are always finite, e.g.
\begin{align}
\label{eq:tracedet}
\operatorname{Tr}(S)=4R_0, \quad \operatorname{det}(S)=R_0^4.
\end{align}
The Ricci scalars are constants and the Kretschmann scalars diverge as $r^{-6}$ 
when $r\rightarrow 0$, if $\rhg$ and $\rhf$ are non-zero.

Both metrics can be diagonalized (not simultaneously) to assume the usual Schwarzschild--anti--de Sitter/Schwarzschild--de Sitter (SAdS/SdS) form, by the general coordinate transformations, for $g_{\mu\nu}$,
\begin{subequations}
\label{eq:changecoord}
\begin{align}
\dd v&=\dd \tilde{t}+G^{-1}\dd \tilde{r} \\[1mm]
\dd r&=\dd\tilde{r},
\end{align}
\end{subequations}
and for $f_{\mu\nu}$,
\begin{subequations}
\begin{align}
\dd v&=R_0^{-1}\left(\dd \hat{t}+F^{-1}\dd\hat{r}\right) \\[1mm]
\dd r&=R_0^{-1}\dd\hat{r}.
\end{align}
\end{subequations}

If we do not take into account the topological constraint of \autoref{subsec:HR}, we can have many possible metric combinations arising from \eqref{eq:examplevacuumg} and \eqref{eq:examplevacuumf}. The horizon radii $\rhg$ and $\rhf$ 
are integration constants of the solutions, as in GR. If $\rhg \neq 0$ and $\rhf=0$, we 
will have the SAdS/SdS solution in $g_{\mu\nu}$ and the 
anti-de Sitter/de Sitter (AdS/dS) in $f_{\mu\nu}$, which do not share their isometry 
groups. AdS and dS are both maximally symmetric \cite{Moschella2006}, whereas SAdS and SdS are not. For $\rhg=\rhf=0$, we can still 
have AdS in one sector and dS in the other, depending on the values of 
$\Lambda_f$ and $\Lambda_g$, again implying different isometry groups. Also, 
one can fix two of the $\beta$ parameters, say $\beta_0$ and 
$\beta_4$, to set $\Lambda _g=\Lambda_f=0$. In that case, we have two 
Schwarzschild solutions having the same isometry groups, but different KVFs. The KVFs for some of these solutions are explicitly computed in a Wolfram Mathematica 11 \cite{mathematica} notebook, which is attached to the paper and available at \href{https://tinyurl.com/y9r7nuuh}{this link (click here)}.\footnote{Explicitly, the link is \href{https://tinyurl.com/y9r7nuuh}{https://tinyurl.com/y9r7nuuh}.}

All these combinations satisfy the bimetric field equations. However, we must also take into account the spacetime topology, as explained in \autoref{subsec:HR}. Since we use the same atlas for both sectors, they must have a common topology. To be precise, our Eddington--Finkelstein chart $(v,r,\theta,\phi)$ does not cover the whole spacetime. However, since the Penrose diagrams of all the possible solutions described by \eqref{eq:examplevacuumg}-\eqref{eq:examplevacuumf} are known \cite{Hawking:1973uf,PhysRevD.15.2738,Charmousis2011}, we know their topology. Hence, we can build an atlas covering the whole spacetime, uniquely determining the topology. As a result, the remaining possible solutions are written in \autoref{tab:allowedmetrics}.
\begin{table*}[t]
\renewcommand{\arraystretch}{1.3}
\setlength{\tabcolsep}{8pt}
	\centering
	\begin{tabular}{P{0.1\textwidth}P{0.1\textwidth}!{\color{white!30!gray}\vrule}P{0.05\textwidth}P{0.05\textwidth}P{0.05\textwidth}P{0.05\textwidth}P{0.05\textwidth}P{0.05\textwidth}P{0.05\textwidth}P{0.05\textwidth}}
	\arrayrulecolor{black}\toprule
 	\multicolumn{2}{c!{\color{white!30!gray}\vrule}}{
 	\diagbox[
 	linewidth=0.075mm,
 	linecolor=white!30!gray,
 	innerwidth=40mm,
 	innerleftsep=13mm, innerrightsep=13mm,
 	outerleftsep=0mm, outerrightsep=0mm]
 	{$g_{\mu\nu}$}{$f_{\mu\nu}$}
 	} & M & S & AdS & dS & SAdS & SdS  \\
	\arrayrulecolor{white!30!gray}\hline
	\rule{0pt}{5mm}
	$\mathbb{R}^4$ & M & \checked & $\times$ & $\bullet$ & $\times$ & $\times$ & $\times$ \\
	$\mathbb{R}^2\times\mathbb{S}^2$ & S & $\times$ & \checked & $\times$ & $\times$ & $\bullet$ & $\times$ \\
	$\mathbb{S}^1\times \mathbb{R}^3$ & AdS & $\bullet$ & $\times$ & \checked & $\times$ & $\times$ & $\times$  \\
	$\mathbb{R}^1\times \mathbb{S}^3$ & dS & $\times$ & $\times$ & $\times$ & \checked & $\times$ & $\times$ \\
	$\mathbb{S}^1\times \mathbb{R}^1 \times \mathbb{S}^2$ & SAdS & $\times$ & $\bullet$ & $\times$ & $\times$ & \checked & $\times$ \\
	$\mathbb{R}^1\times \mathbb{S}^1 \times \mathbb{S}^2$ & SdS & $\times$ & $\times$ & $\times$ & $\times$ & $\times$ & \checked \\
	\arrayrulecolor{white!30!gray}\cline{1-2}
	Topology 	& Solution &  &  &  &  &  & \\
	\arrayrulecolor{black}\midrule \\[-3.5mm]
	\multicolumn{8}{l}{Comments:} \\[-3.5mm]
	\multicolumn{8}{p{0.75\textwidth}}{\footnotesize
		\begin{enumerate}
			\itemsep=0em
			\item ``M" stands for Minkowski and ``S" stands for Schwarzschild.
			\item For AdS and SAdS, one can consider the universal covering in the timelike direction in order to avoid violation of causality, ``unwrapping" $\mathbb{S}^1$ into $\mathbb{R}^1$. This would eliminate the closed timelike curves present in AdS and SAdS, and would assign them the topologies, respectively, $\mathbb{R}^4$ and $\mathbb{R}^2\times \mathbb{S}^2$ \cite{Charmousis2011,Hawking:1973uf}.
			\item $\bullet$ stands for \checked \hspace{1mm}if we take the universal coverings of AdS and SAdS, having topologies $\mathbb{R}^4$ and $\mathbb{R}^2\times \mathbb{S}^2$ (see next comment), otherwise it stands for $\times$.
			\item $\mathbb{S}^2$ in the spatial topology of Schwarzschild, SAdS and SdS is necessary due to the presence of an event horizon, for stationary black holes (then, it holds also for static black holes) in 4 dimensions \cite[p.~323]{Wald:1984rg},\cite{Hawking1972}.
			\item Note that the topologies can be matched in some cases. For example, one can have a punctured Minkowski  with $\mathbb{R}^4\backslash\{r=0\}$ rather than M in the M and S case.
		\end{enumerate}			
	} \\
	\arrayrulecolor{black}\bottomrule
	\end{tabular}
	\caption{Metric combinations for the solution \eqref{eq:examplevacuumg}--\eqref{eq:examplevacuumf}, which are allowed and forbidden by the topological constraint.}
	\label{tab:allowedmetrics}
\end{table*}

We now discuss the relation between the KVFs and the isometry groups of the possible solutions. If $g_{\mu\nu}$ and $f_{\mu\nu}$ are different, the isometry groups and the KVFs are also different. Consider the case where $g_{\mu\nu}$ and $f_{\mu\nu}$ belong to the same ``class"--- i.e., they have the same form, but different numerical values of their parameters. First, when $0\neq \Lambda_g \neq \Lambda_f \neq 0$, the KVFs of $g_{\mu\nu}$ and $f_{\mu\nu}$ are different, because they depend on the specific values of the cosmological constants (see e.g. \cite{Henneaux1985} and \cite{Jamal2017} for the explicit expressions of the KVFs of, respectively, AdS and dS). However, their isometry group is the same. Second, when $\Lambda_g =\Lambda _f=0$, we have the bi-Schwarzschild or the bi-Minkowski solution. In this case, the isometry groups are again the same, but the KVFs are different and related by a diffeomorphism $\varphi$. If $\xi^\mu$ and $\eta^\mu$ are KVFs of $f_{\mu\nu}$ and $g_{\mu\nu}$, respectively, we have
\begin{align}
\label{eq:pushforward}
\xi^\mu(\varphi^\beta(x^\alpha))=\left(\varphi_*\right){}^\mu{}_\nu\eta^\nu(x^{\alpha}),
\end{align}
with $\varphi$ being the diffeomorphism. This means that the components of the KVFs are different in the same chart $x^\alpha=(v,r,\theta,\phi)$. We will return to \eqref{eq:pushforward} in \autoref{sec:geometry}.

\subsubsection*{Example II: Bi-Minkowski with ``screened" matter sources}

The relevance of this example resides in having metrics with the same isometry group, but different KVFs. Also, it shows the possibility of having nongravitational matter screened by the ``effective stress--energy tensors" $V_{g}{}^\mu{}_\nu$ and  $V_{f}{}^\mu{}_\nu$.

This example again exhibits a diffeomorphism between the $g$-sector and the $f$-sector, analogously to \eqref{eq:pushforward}. We consider two Minkowski metrics in a chart $(t,r,\theta,\phi)$ which constitutes the usual spherical polar chart for $g_{\mu\nu}$, but not for $f_{\mu\nu}$. In particular, we use the following bidiagonal ansatz,
\begin{align}
\label{eq:doubleMinkowski}
g_{\mu\nu}&=
	 \begin{pmatrix}
	   	-1 & 0 & 0 & 0 \\
	   	0 & 1 & 0 & 0 \\
	   	0 & 0 & r^2 & 0 \\
	   	0 & 0 & 0 & r^2 \sin^2 (\theta ) \\
	  \end{pmatrix}\!, \nonumber \\
f_{\mu\nu}&=
	 \begin{pmatrix}
	   	-\tau^{\prime\,2}(t) & 0 & 0 & 0 \\
	   	0 & \rho^{\prime\,2}(r) & 0 & 0 \\
	    0 & 0 & \rho^2(r) & 0 \\
	   	0 & 0 & 0 & \rho^2(r) \sin^2 (\theta ) \\
	  \end{pmatrix}\!, \nonumber \\
S{}^\mu{}_\nu&=
	 \begin{pmatrix}
	   	\tau'(t) & 0 & 0 & 0 \\
	   	0 & \rho'(r) & 0 & 0 \\
	    0 & 0 & \rho(r)/r & 0 \\
	   	0 & 0 & 0 & \rho(r)/r \\
	  \end{pmatrix}\!,
\end{align}
where $S{}^\mu{}_\nu$ is the principal square root of $g^{-1}f$ and the prime means ``derivative with respect to the argument."

The trace and the determinant of $S{}^\mu{}_\nu$ are,
\begin{align}
\label{eq:tracedet2}
\operatorname{Tr}\left(S\right)&=\dfrac{2\rho(r)}{r}+\rho'(r)+\tau'(t), \nonumber \\ \mydet{S}&=\dfrac{\rho^2(r)}{r^2}\rho'(r)\tau'(t),
\end{align}
so assuming $\tau'(t),\rho'(r)>0$, the square root is invertible and the pathologies described in Proposition 1 in \cite[Section 3.1]{PhysRevD.96.064003} will not appear.

The change of coordinates which will put $f_{\mu\nu}$ in the standard Minkowski form is,
\begin{align}
\label{eq:genericChangeCoord}
\dd \tau=\tau'(t)\dd t, \quad \dd \rho=\rho'(r)\dd r.
\end{align}

With the ansatz \eqref{eq:doubleMinkowski}, the Bianchi constraints in \eqref{eq:Bianchi} are satisfied for every $\tau (t)$ and $\rho(r)$. In vacuum, the field equations \eqref{eq:EoM} are reduced to,
\begin{align}
\label{eq:almostT}
\frac{m^4}{M_{g}^2}V_{g}{}^\mu{}_\nu & =0,\quad \frac{m^4}{M_{f}^2}V_{f}{}^\mu{}_\nu=0,
\end{align}
which is true only for zero tensor potentials $V_{g,f}{}^\mu{}_\nu$. However, we can couple two different generic stress--energy tensors $T_g{}^\mu{}_\nu$ and $T_f{}^\mu{}_\nu$ to the metrics (see, e.g., \cite{deRham:2014naa}),
\begin{align}
\label{eq:T}
\frac{m^4}{M_{g}^2}V_{g}{}^\mu{}_\nu & =\frac{1}{M_{g}^2} T_g{}^\mu{}_\nu,\quad \frac{m^4}{M_{f}^2}V_{f}{}^\mu{}_\nu=\frac{1}{M_{f}^2} T_f{}^\mu{}_\nu.
\end{align}
and solve \eqref{eq:T} for $T_g{}^\mu{}_\nu$ and $T_f{}^\mu{}_\nu$. Note that the Bianchi constraints automatically imply that the stress--energy tensors defined in this way are divergence-less. In addition, the stress--energy tensors always have an energy density $T_g{}^0{}_0(r)$ independent of time, and a time-dependent isotropic non-homogeneous pressure, $T_g{}^1{}_1(t,r)\neq T_g{}^2{}_2(t,r)=T_g{}^3{}_3(t,r)$.

The field equations are thus formally satisfied, and the two sectors share their isometry group, even if they have different KVFs. Indeed, \eqref{eq:pushforward} holds with the following Jacobian matrix,
\begin{align}
\label{eq:iJac}
\left(\varphi_*\right){}^\mu{}_\nu=\left(J^{-1}\right){}^\mu{}_\nu&=
	 \begin{pmatrix}
	   	\tau'(t)^{-1} & 0 & 0 & 0 \\
	   	0 & \rho'(r)^{-1} & 0 & 0 \\
	    0 & 0 & 1 & 0 \\
	   	0 & 0 & 0 & 1 \\
	  \end{pmatrix}\!.
\end{align}

Two things should be remarked concerning this solution. First, the two matter sources do not couple with the Einstein tensors of the metrics, but rather with the tensors potentials $V_{g}{}^\mu{}_\nu$ and $V_{f}{}^\mu{}_\nu$. This does not mean that we are decoupling $g_{\mu\nu}$ and $f_{\mu\nu}$, retaining two copies of GR, because the Bianchi constraint \eqref{eq:Bianchi} is an additional requirement not present in GR. Therefore we cannot choose any two metrics. Second, in this peculiar set up we can have matter gravitationally decoupled from us, since its contribution is exactly canceled by tensor potentials.

We can require the stress--energy tensors not to diverge at radial and time infinity or at finite values of $t$ and $r$. This can be achieved by setting, for example,
\begin{align}
\label{eq:taurho}
\tau(t)&=t+\operatorname{arcsinh}(t), \, t\in \mathbb{R} \nonumber \\
\rho(r)&=r+\operatorname{arcsinh}(r), \, r>0.
\end{align}
This is a valid diffeomorphism from the Minkowski spacetime into itself, and the resulting stress--energy tensors are always finite and satisfy,
\begin{align}
\label{eq:stresslimits}
{T_g{}^0}_0(r=0)&=\rho^g_0<\infty, \nonumber \\
{T_f{}^0}_0(r=0)&=\rho^f_0<\infty, \nonumber \\
\lim\limits_{r\rightarrow\infty}{T_g{}^0}_0(r)&=\rho^g_\infty<\infty, \nonumber \\
\lim\limits_{r\rightarrow\infty}{T_f{}^0}_0(r)&=\rho^f_\infty<\infty.
\end{align}
Analogous relations are valid for the other components of the stress--energy tensors, even if they depend on time.

The metrics $g_{\mu\nu}$ and $f_{\mu\nu}$ share the SO(3) KVFs, since the coordinate transformation between 
the sectors does not involve the angular coordinates [see \eqref{eq:iJac}]. 
As we will see in Proposition 2 in \autoref{subsec:isometriesHR}, since the Lie 
derivatives of $S{}^\mu{}_\nu$ with respect to the SO(3) KVFs 
are zero, the SO(3) KVFs are collineations for all the tensors in this 
solution. However, the Lie derivatives of $S{}^\mu{}_\nu,V_g{}^\mu{}_\nu,V_f{}^\mu{}_\nu$ with respect to all other 
KVFs of $g_{\mu\nu}$ and $f_{\mu\nu}$ are non-zero, so these KVFs are not 
collineations for the whole system.

\subsubsection*{Example III: Einstein solutions with algebraically decoupled parameters}

With this example we want to show that requiring a nonsingular geometry can constrain the KVFs of the metrics. As in the previous example, we consider the same isometry group, but different KVFs.

If we assume $g_{\mu\nu}$ to be an Einstein metric, then the field equations in the $g$-sector are just $V_g{}^\mu{}_\nu=0$. As explained in \cite{Kocic:2017wwf}, this equation determines the (eigenvalues of the) square root. There is a specific set of $\beta$-parameters for which $V_g{}^\mu{}_\nu=0$ is satisfied with some of the square root eigenvalues $\lambda_i$ left undetermined. We call this case ``algebraically decoupled". Note that, even if we assume $g_{\mu\nu}$ to be Einstein, $f_{\mu\nu}$ is not arbitrary. This is due to the fact that, in HR theory, $f=g\,S^2$ even if the metrics are dynamically decoupled. Hence, imposing $V_g{}^\mu{}_\nu=0$ already determines $f_{\mu\nu}$ (or equivalently $S^\mu{}_\nu$), and makes the theory not equivalent to two separate copies of GR.

One algebraically decoupled case is when the two metrics are bidiagonal (Type I in \cite{Hassan:2017ugh}) and
\begin{align}
(\beta_1 \beta_2 - \beta_0 \beta_3)^2 &= 4(\beta_1^2 - \beta_0 \beta_3)(\beta_2^2 - \beta_1 \beta_3), \nonumber\\
\beta_2^2-\beta_1 \beta_3 &\neq 0.
\end{align}
In this case, choosing the eigenvalues of the square root matrix to be
\begin{align}
	\lambda_1 = \lambda_2 = \lambda_3 = -\frac{\beta_1 \beta_2 - \beta_0 \beta_3}{2(\beta_2^2 - \beta_1 \beta_3)},
\end{align}
solves $V_g{}^\mu{}_\nu=0$ with the last eigenvalue, $\lambda_4(x)$, left unspecified. The equation $V_f{}^\mu{}_\nu=0$ becomes an equation in the $\beta$-parameters and can be solved for $\beta_4$.

Let $g_{\mu\nu}$ be a specific Einstein metric; for definiteness we let $g_{\mu\nu}$ be the Schwarzschild metric,
\begin{align}
\label{eq:Schw}
	g_{\mu\nu}=-F\mathrm{d}t^2+F^{-1}\mathrm{d}r^2 +r^2\mathrm{d}\Omega^2,
\end{align}
where $\mathrm{d}\Omega^2=\mathrm{d}\theta^2+\sin^2 \theta \mathrm{d}\phi^2$ and $F(r):=1-r_\mathrm{H}/r$, rendering the corresponding Einstein tensor to vanish. The only equation left to satisfy is $G_f{}^\mu{}_\nu=0$, which becomes a set of differential equations for $\lambda _4(x)$. A set of solutions can be found through a separation of variables ansatz. In particular, for such an ansatz, the arbitrary function $\lambda _4(x)=\lambda (\phi)$ solves the Einstein equation for $f_{\mu\nu}$. Hence the metric \eqref{eq:Schw} together with
\begin{align}
	f_{\mu\nu}=\lambda_1^2\left(-F\mathrm{d}t^2+F^{-1}\mathrm{d}r^2 +r^2\mathrm{d}\theta^2\right)+ \lambda^2(\phi)r^2\sin^2 \theta \mathrm{d}\phi^2
\end{align}
solve the bimetric field equations, with $\lambda(\phi)$ being an arbitrary function and with parameters specified above.
Assuming that $\lambda(\phi)>0$, we can make the change of coordinates,
\begin{subequations}
\begin{align}
\phi &\rightarrow \Phi, \quad \mathrm{d}\Phi = \lambda(\phi)\mathrm{d}\phi \\
t&\rightarrow T = \lambda_1 t  \\
r &\rightarrow R = \lambda_1 r  \\
\theta &\rightarrow \Theta=\theta
\end{align}
\end{subequations}
where $\lambda(\phi)=\Phi^\prime(\phi)$ and $\Phi(\phi)$ is a monotonic function (and thus invertible). In the new chart, $f_{\mu\nu}$ is manifestly Schwarzschild,
\begin{align}
f_{\mu\nu}	&=-\left(1-\frac{R_H}{R}\right)\mathrm{d}T^2+\frac{\mathrm{d}R^2}{1-\frac{R_H}{R}}  \nonumber \\
				&\quad \, +R^2\left(\mathrm{d}\Theta^2+ \sin^2 \Theta \mathrm{d}\Phi^2\right)
\end{align}
where we defined $R_H:=\lambda_1 r_h$. Thus, $f_{\mu\nu}$ also being a Schwarzschild metric, $g_{\mu\nu}$ and $f_{\mu\nu}$ exhibit the same isometry group and determine the same topology (see \autoref{subsec:HR}). The KVFs of $g_{\mu\nu}$ are 
\begin{subequations}
\label{eq:KVFg}
\begin{align}
	\eta_0 &= \partial_t, \quad \eta_1 = \partial_\phi,  \\
	\eta_2 &= \cos \phi \,\partial_\theta -\cot \theta \sin \phi \,\partial_\phi, \\
	\eta_3 &= - \sin \phi \,\partial_\theta - \cot \theta \cos \phi \,\partial_\phi
\end{align}
\end{subequations}
and the related KVFs of $f_{\mu\nu}$ are
\begin{subequations}
\label{eq:KVFf}
\begin{align}
	\xi_{0} & =\partial_{T}=\frac{1}{\lambda_{1}}\partial_{t},\quad\xi_{1}=\partial_{\Phi}=\frac{\lambda_{1}}{\Phi^{\prime}(\phi)}\partial_{\phi} \\
	\xi_{2} & =\cos(\Phi(\phi))\partial_{\theta}-\frac{\lambda_{1}}{\Phi^{\prime}(\phi)}\cot\theta\sin(\Phi(\phi))\partial_{\phi}, \\
	\xi_{3} & =-\sin(\Phi(\phi))\partial_{\theta}-\frac{\lambda_{1}}{\Phi^{\prime}(\phi)}\cot\theta\cos(\Phi(\phi))\partial_{\phi}
\end{align}
\end{subequations}
If $\lambda(\phi)$ is non-constant, only the KVF generating time translations is a KVF for both metrics.

A final check reveals that the scalar invariants of the square root $S=\mathrm{diag}(\lambda_1,\lambda_1,\lambda_1,\lambda_4(\phi))$ are all non-zero and finite, assuming that $\lambda_1,\lambda_4(\phi)>0$, i.e., the principal square root branch.

Contracting the KVFs \eqref{eq:KVFg} and \eqref{eq:KVFf} in different combinations with the metrics yields a set of scalar fields, which must be periodic in $\phi$ so that $\phi=0$ and $\phi=2\pi$ give the same value. Demanding this to be the case is equivalent to $\lambda(\phi)$ being periodic in $\phi$, i.e., $\lambda(\phi)=\lambda(\phi+2\pi)$, and $\Phi(\phi+2\pi)=\Phi(\phi)+2\pi n$ where $n$ is an integer. However, $\Phi(\phi)$ is a monotonic (increasing) function of $\phi$ so $n$ must be positive.  Conversely, we must demand that if we invert the relation between the azimuth coordinates to obtain $\phi=\phi(\Phi)$, the scalars must be periodic in $\Phi$ with period $2\pi$. This is equivalent to $\phi(\Phi=2\pi)=\phi(\Phi=0)+2\pi n'$ where $n'$ is a positive integer.

Note that we imposed $\Phi \in \left[0,2\pi \right)$ since $\Phi $ is the azimuthal coordinate in the spherical polar chart of $f_{\mu\nu}$. If $\Phi \in \left[ 0, \Phi_0\right)$ with $\Phi_0<2\pi$, we could identify the hypersurfaces $\Phi=0 $ and $\Phi =\Phi_0$ by demanding $\lambda$ to have a period of $\Phi_0$, but this would lead to a conical singularity in the $f$-sector, as explained in \cite[p.~214]{Wald:1984rg}. On the other hand, if $\Phi \in \left[ 0, \Phi_0\right)$ with $\Phi_0>2\pi$, we would cover twice the region between $2\pi<\Phi<\Phi_0$. Then, we could restrict $\Phi \in \left[0,2\pi \right)$, but $\Phi$ and $\phi$ both being monotonic would lead to $\phi \in \left[ 0, \phi_0\right)$ with $\phi_0<2\pi$, i.e., a conical singularity in the $g$-sector. Introducing the conical singularity does not change the topology (since $r=0$ is already not part of the spacetime); however, the solution would not be Schwarzschild.

As an example, consider the functions $\Phi(\phi) =\phi + \frac{1}{2}\sin \phi$ and $\lambda_4(\phi) = \Phi^\prime(\phi) = 1 + \frac{1}{2}\cos \phi$. The function $\lambda_4(\phi)$ is strictly positive and periodic in $\phi$ with period $2\pi$. Hence, $\Phi(\phi)\in \left[ 0,2 \pi \right)$ is a monotonically increasing function of $\phi$, satisfying the above conditions. This function determines different KVFs for the metrics, according to \eqref{eq:KVFf}.
Another example is when $\lambda(\phi)=c$ with $c$ being a positive constant. In this case, $\lambda(\phi)$ is clearly periodic in $\phi$ and $\Phi=c\phi$. Concerning the property $\Phi(2\pi)=\Phi(0)+2\pi n$, this demands that $c=n$ is a positive integer and $\phi(2\pi)=\phi(0)+2\pi n'$ implies that $1/c=n'$ is a positive integer. 
Hence $c=1$, as discussed in \cite{Kocic:2017wwf}. 

Note that this method of generating solutions did not rely on $g_{\mu\nu}$ being the Schwarzschild metric. The method generalizes straightforwardly to other Einstein metrics and other sets of $\beta$-parameters which have the property of leaving one or more of the square root eigenvalues $\lambda_i$ undetermined. 

\subsection{Collineations in the Hassan-Rosen bimetric theory}
\label{subsec:isometriesHR}

In this subsection we study general properties of spacetime symmetries in HR bimetric theory.

We start by reminding that the KVFs of a metric $g_{\mu\nu}$ are defined to be the solutions of the Killing equation,
\begin{align}
\label{eq:Kequation}
\mathscr{L}_{\eta}g_{\mu\nu}=\nabla_{(\mu} \eta _{\nu)} =0,
\end{align}
with $\mathscr{L}_\eta$ the Lie derivative along the vector field $\eta ^\mu$ and $\nabla _\mu$ the covariant derivative compatible with $g_{\mu\nu}$. The parentheses are understood as the symmetrisation of the indices they enclose.

The KVFs are completely specified by \cite{Wald:1984rg}
\begin{subequations}
\label{eq:KVFs}
\begin{align}
v^\rho\nabla_\rho\eta_\mu &=v^\rho L_{\rho\mu}, \label{eq:KVFsa}\\
v^\rho\nabla_\rho L_{\mu\nu} &=-v^\rho R_{\mu\nu\rho}{}^\sigma\eta_\sigma,\label{eq:KVFsb}
\end{align}
\end{subequations}
where $L_{\mu\nu}\coloneqq \nabla_\mu\eta_\nu$ is antisymmetric due to the Killing equation $L_{\mu\nu}=-L_{\nu\mu}$. The non-zero vector field $v^\mu$ specifies the integration path along which the differential 
system is solved. As shown in \cite{Wald:1984rg}, one can make use of \eqref{eq:KVFsb} to write
\begin{align}
\label{eq:GRLaplacian}
\dalamb\eta_\mu&=\nabla_\nu L{}^\nu{}_\mu=-R{}^\rho {}_{\mu\rho}{}^\nu\eta_\nu =-R{}_\mu{}^\nu\eta_\nu,
\end{align}
where $\dalamb \coloneqq \nabla_\nu \nabla^\nu$.
We can raise the index $\mu$ to obtain,
\begin{align}
\label{eq:HRLaplacianVectors}
\dalamb\eta^\mu=-R{}^\mu{}_\nu\eta^\nu .
\end{align}

\paragraph{Spacetime symmetries and field equations.}
A known result in GR is that a KVF is also a collineation for the Einstein tensor, the Riemann tensor and the stress-energy tensor \cite{hall2004symmetries},
\begin{align}
\mathscr{L}_{\eta }\,g{}_{\mu\nu}=0	&\Longrightarrow \mathscr{L}_{\eta }\,R{}_{\mu\nu\alpha}{}^\beta=0 \nonumber \\
									&\Longrightarrow \mathscr{L}_{\eta }\,G{}_{\mu\nu}=0 \nonumber \\
									&\Longrightarrow \mathscr{L}_{\eta }\,T{}_{\mu\nu}=0.
\end{align}
The first two results can be straightforwardly extended to the HR bimetric theory, whereas the third result can not since it uses the Einstein equations.

Consider the KVF $\eta^\mu$ of the metric $g_{\mu\nu}$ (analogous computations can be performed for the $f$-sector). The Lie derivative of the field equations \eqref{eq:EoMg} with respect to $\eta^\mu$ is,
\begin{align}
\label{eq:CollineationFE}
M_{g}^2\,\mathscr{L}_{\eta}\,G_g{}^\mu{}_\nu+m^4\,\mathscr{L}_{\eta}\,V_g{}{}^\mu{}_\nu=\mathscr{L}_{\eta}\,T_g{}^\mu{}_\nu.
\end{align}
We know that $\mathscr{L}_{\eta}\,G_g{}^\mu{}_\nu=0$, hence,
\begin{align}
\label{eq:CollineationFE2}
m^4\mathscr{L}_{\eta}\,V_g{}^\mu{}_\nu=\mathscr{L}_{\eta}\,T_g{}^\mu{}_\nu.
\end{align}
Therefore, in general, a KVF of one metric is neither a collineation of the respective stress--energy tensor nor a collineation of its tensor potential (we showed this explicitly in \autoref{subsec:explicit}). Also, this opens up the possibility of finding solutions having a given KVF for one metric but different collineations for both $V_g{}^\mu{}_\nu$ and $T_g{}^\mu{}_\nu$. As already mentioned, some solutions having this property were presented in \cite{PhysRevD.86.061502,Nersisyan:2015oha}. Therefore, we notice that the \emph{definition} of a spacetime symmetry in HR bimetric theory is non-trivial. Suppose having a static and spherically symmetric $g_{\mu\nu}$ and an axially symmetric $T_g{}^\mu{}_\nu$ and $V_g{}^\mu{}_\nu$. Whether this should be considered a spherically symmetric system or an axially symmetric one remains an open question. However, in ``bimetric vacuum" (see \cite{Kocic:2017hve} for a discussion about vacuum in HR bimetric theory), \eqref{eq:CollineationFE2} becomes,
\begin{align}
\label{eq:CollineationV}
\mathscr{L}_{\eta}\,V_g{}^\mu{}_\nu=0,
\end{align}
and the KVFs of the metric are also collineations for the corresponding tensor potential.

In the following we analyze some properties of the KVFs, useful for determining them in HR bimetric theory. First, \eqref{eq:KVFs} does not make use of the Einstein equations, so it holds for $g_{\mu\nu}$ and $f_{\mu\nu}$ separately.  However, the difference with GR is not only that we have separate differential systems for each metric. Inserting the bimetric field equations \eqref{eq:EoM2} in \eqref{eq:GRLaplacian}, we get,
\begin{subequations}
\begin{align}
\label{eq:HRLaplacian}
&-\dalamb\eta_\mu	=R^g{}_\mu{}^\nu\eta_\nu \nonumber \\
				&=\frac{1}{M_{g}^2}\left[m^4\left(\dfrac{1}{2}\,V_g\,\delta {}_\mu{}^\nu-V_{g}{}_\mu{}^\nu \right)+T_g{}_\mu{}^\nu-\dfrac{1}{2}T_g\delta{}_\mu{}^\nu\right]\eta_\nu, \\[2mm]
&-\tilde{\nabla}^2\xi_\mu	=R^f{}_\mu{}^\nu\xi_\nu \nonumber \\
						&=\frac{1}{M_{f}^2}\left[m^4\left(\dfrac{1}{2}\,V_f\,\delta {}_\mu{}^\nu-V_{f}{}_\mu{}^\nu \right)+T_f{}_\mu{}^\nu-\dfrac{1}{2}T_f\delta{}_\mu{}^\nu\right]\xi_d.
\end{align}
\end{subequations}
Therefore in vacuum, contrary to GR, the Laplacian of a KVF is not necessarily zero because of the contributions from the tensor potentials [see \eqref{eq:GRLaplacian}]. This reflects the more complicated structure of the theory, and makes the search for KVFs for vacuum solutions more difficult. However, for a generic vector field $V^\mu$ and independently of any field equations, the ``Komar identity" holds \cite{PhysRev.113.934,doi:10.1063/1.1705011}
\begin{equation}
\label{eq:Komar}
\nabla _\mu \nabla _\nu \nabla ^{[\nu}V^{\mu]}\equiv 0,
\end{equation}
which, for a KVF, becomes
\begin{equation}
\label{eq:DivLap}
\nabla_\mu \left( \dalamb \eta^\mu \right)\equiv 0,
\end{equation}
because $\nabla ^{[\nu}\eta^{\mu]}\coloneqq \left( \nabla^\nu\eta^\mu-\nabla^\mu\eta^\nu \right)/2=\nabla ^\nu \eta^{\mu}$ due to the Killing equation \eqref{eq:Kequation}.
One could solve \eqref{eq:DivLap} for KVFs in the same way one would use \eqref{eq:HRLaplacianVectors} in GR. An analogous relation holds for the $f$-sector. An alternative proof for \eqref{eq:Komar}, simpler than the one in \cite{doi:10.1063/1.1705011}, is provided in appendix \ref{app:alternativeproof}.

So far, we have not assumed anything regarding the KVFs. Now we assume that the two metrics share  a KVF $X^\mu$. Then, substituting \eqref{eq:EoM2}, \eqref{eq:potentials} and \eqref{eq:potentialsTrace} in \eqref{eq:HRLaplacian} and performing some algebra one has,
\begin{align}
\label{eq:SameKVF}
M_g^2\,\dalamb X^\mu+M_f^2\,\mydet{S}\tilde{\nabla}^2X^\mu&=-m^4\,V(S)X^\mu.
\end{align}
Fulfilling \eqref{eq:SameKVF} is then a necessary condition for $X^\mu$ to be a KVF for both metrics. It is an additional constraint to \eqref{eq:DivLap} when solving for shared KVFs.

\paragraph{Conserved currents determined by KVFs.}
Now consider the vector current $V_g^{\mu\nu}\eta_\nu$, where $\eta^ \mu$ is a KVF of $g_{\mu\nu}$. Its divergence is zero,
\begin{align}
\label{eq:current}
\nabla_\mu \left(V_g^{\mu\nu}\eta_\nu \right)=\eta_\nu\nabla_\mu V_g^{\mu\nu}+V_g^{\mu\nu}\nabla_\mu \eta_\nu =0,
\end{align}
and therefore it is a conserved current. The last equality follows from the Bianchi constraint for $V_g{}^{\mu\nu}$, the symmetry of $V_g{}^{\mu\nu}$ and the antisymmetry of $L_{\mu\nu}=\nabla_\mu \eta_\nu$ (the same result can be stated for the $f$-sector). This conservation law is the same as that for the stress--energy tensor
\begin{align}
\nabla_\mu \left(T_g^{\mu\nu}\eta_\nu \right)=0,
\end{align}
which is also true in GR.\footnote{Their validity is shown in the same way as for the current $V_g{}^{\mu\nu}\eta_{\nu}$.} Considering the $g$-sector, we can rewrite \eqref{eq:current} in the following way,
\begin{align}
\label{eq:conservation}
\partial _\mu \left( \sqrt{-g}V_g^{\mu\nu}\eta_{\nu} \right)=0.
\end{align}
Separating the time and spatial derivatives and integrating over a spacelike hypersurface $\mathcal{V}$ defined by $x^0=\mathrm{const.}$, we obtain,
\begin{align}
\label{eq:integral}
\partial _0 \int_\mathcal{V} \dd ^3x\left( \sqrt{-g}V_g^{0\nu}\eta_{\nu} \right)	&=-\int_\mathcal{V} \dd ^3x \, \partial _i \left( \sqrt{-g}V_g^{i\nu}\eta_{\nu} \right) \nonumber \\
&=-\int_{\partial \mathcal{V}} \dd \sigma_i \sqrt{-g}V_g^{i\nu}\eta_{\nu},
\end{align}
where $\dd\sigma_i$ is the two-surface element on $\partial \mathcal{V}$.
When $V_g ^{\mu\nu}$ tend to zero at spatial infinity, e.g., in asymptotically flat spacetimes with two Minkowski metrics, we can define the conserved charge,\footnote{In HR bimetric theory, asymptotically flat spacetimes do not necessarily have two Minkowski metrics with the same components \cite{Kocic:2017wwf}.}
\begin{align}
\label{eq:conservedscalar}
V^0 \coloneqq \int_\mathcal{V} \dd ^3x\left( \sqrt{-g}V_g^{0\nu}\eta_{\nu} \right), \quad \partial _0 V^0=0.
\end{align}
Regarding the meaning of $V^0$, we can consider an asymptotically flat static spacetime with a stress--energy tensor. Then, both $V^0$ and $P^0 \coloneqq \int_\mathcal{V} \dd ^3x\left( \sqrt{-g}\,T^{0\nu}\eta_{\nu} \right)$ would be separately conserved. Therefore the total energy of the system would be the sum of the two, i.e.
\begin{align}
\label{eq:energy}
E=P^0+V^0.
\end{align}
Hence, in this case $m^4\, V_g^{00}$ can be interpreted as the energy density due to the tensor potential.

\paragraph{Structure of spacetime symmetries.}
We now state several claims which clarify the behavior of the isometries and can help to define a symmetric spacetime in the HR bimetric theory.

We start by stating the following (in matrix notation)
\begin{lem*}
  Let $\xi$ be a vector field, $S$ a \emph{(1,1)}-tensor field and 
  $F$ a matrix-valued function. If the Lie
  derivative of $S$ with respect to $\xi$ vanishes, then the Lie derivative
  of $F(S)$ vanishes too; that is, $\mathscr{L}_{\xi}S=0$ implies
  $\mathscr{L}_{\xi}F(S)=0$.
\end{lem*}
\begin{proof}
  Locally, one can always find a chart where a coordinate $x^0$ is aligned
  along $\xi^\mu$ so that $\xi^a = \delta^a_0$,
\begin{align}
  \mathscr{L}_\xi S^\mu{}_\nu &=  \xi^\rho \partial_\rho S^\mu{}_\nu
  - (\partial_\rho \xi^\mu) S^\rho{}_\nu 
  + (\partial_\nu \xi^\rho) S^\mu{}_\rho \nonumber\\
  & = \partial_0 S^\mu{}_\nu.
\end{align}  
  If $\mathscr{L}_\xi S = 0$, 
  then $S$ is not 
  dependent on $x^0$, since $\partial_0 S(x) =0$. Consequently $F(S)(x)$ does not depend on $x^0$,
  and $\partial_0 F(S)(x)=0$.
\end{proof}
\noindent Thanks to this Lemma, we know that
\begin{align}
\mathscr{L}_\xi \, S{}^\mu{}_\nu =0 \Longrightarrow \mathscr{L}_\xi \, V_{g,f}{}^\mu{}_\nu =0.
\end{align}
\begin{prop}
\label{prop:2}
  If $F$ is invertible then, $\mathscr{L}_{\xi}S=0$ if and only if
  $\mathscr{L}_{\xi}F(S)=0$.
\end{prop}
\begin{proof}
  Since $F$ is invertible, we can use the previous Lemma in both directions.
\end{proof}
\noindent Proposition 1 does not apply to our case, because the tensor potentials, containing traces of $S{}^\mu{}_\nu$, are not invertible functions of it. Therefore
\begin{align}
\mathscr{L}_\xi \, V_{g/f}{}^\mu{}_\nu =0 \centernot\Longrightarrow \mathscr{L}_\xi \, S^\mu{}_\nu =0.
\end{align}
Therefore even in vacuum, as we saw in \autoref{subsec:explicit}, we can have different KVFs in the two sectors. In such a case, the isometry of $g_{\mu\nu}$ are spacetime symmetries for tensors in the $g$-sector and the isometries of $f_{\mu\nu}$ will be spacetime symmetries for tensors in the $f$-sector.

If $g_{\mu\nu}$ and $S{}^\mu{}_\nu$ share their collineations, $f_{\mu\nu}$ shares them too. This is a special case of a more general result we are going to introduce after the following definitions.
\begin{defn}
  Consider metric fields $g$ and $f$. Let $g^{-1}f$ be positive definite.
  We define the one-parameter set of Lorentzian metrics $\Gamma =\{h_{\alpha}=g(g^{-1}f)^{\alpha}$, $\alpha\in \mathbb{R}\}$,
  where the matrix power function is defined through 
  $X^{\alpha}=\exp\left(\alpha\log X\right)$.
  Notice that $g=h_{0}$ and $f=h_{1}$.\footnote{The geometric mean of two symmetric matrices $A$ and $B$ is defined
  	by $A\operatorname{\#}B=A(A^{-1}B)^{1/2}$. For any two $h_{\alpha}$
  	and $h_{\beta}$ we have 
  	$h_{\alpha}\operatorname{\#}h_{\beta}=h_{(\alpha+\beta)/2}$. 
  	See \cite{Hassan:2017ugh} for more details.}
\end{defn}
\begin{defn}
  If the vector field $\xi^\mu$ is a KVFs for $g_{\mu\nu}$ and $f_{\mu\nu}$, we define it as a ``bimetric isometry". The definition extends to any type of spacetime symmetries (e.g., conformal vector fields).\footnote{According to this definition, the conformal vector fields found in \cite{Kocic:2017hve} are bimetric.}
\end{defn}
\begin{defn}
  If $\xi^\mu$ is a spacetime symmetry for all the tensors in one metric sector, then we refer to it as a ``sectoral spacetime symmetry".
\end{defn}
\begin{defn}
  If $\xi^\mu$ is a collineation for some tensors, we call it a ``narrow collineation". The definition extends to any spacetime symmetry, not only collineations.
\end{defn}
\noindent We now state
\begin{prop}
\label{prop:means}
  Consider a vector field $\xi$ and two arbitrary $h_{\alpha}$ and $h_{\beta}$
  such that $\mathscr{L}_{\xi}h_{\alpha}=0$ and $\mathscr{L}_{\xi}h_{\beta}=0$.
  Then $\mathscr{L}_{\xi}h_{\gamma}=0$ for any $\gamma \in \mathbb{R}$.
\end{prop}
\begin{proof}
We define $\mathcal{A}=g^{-1}f$ for readability.

Our hypothesis is
\begin{align}
\label{eq:hyp}
\mathscr{L}_\xi h_\alpha =\mathscr{L}_\xi h_\beta= 0.
\end{align}
Suppose $\alpha < \beta$ (otherwise switch them). We have
\begin{align}
\label{eq:step1}
0=\mathscr{L}_\xi h_\alpha = \left(\mathscr{L}_\xi \,g \right) \mathcal{A}^\alpha+ g\, \mathscr{L}_\xi\mathcal{A}^\alpha,
\end{align}
which is equivalent to
\begin{align}
\label{eq:step2}
\left(\mathscr{L}_\xi \,g \right) \mathcal{A}^\alpha=- g\, \mathscr{L}_\xi \mathcal{A}^\alpha.
\end{align}
From \eqref{eq:hyp} and \eqref{eq:step2} it follows
\begin{align}
\label{eq:step3}
0=\mathscr{L}_\xi h_\beta 	&= \left(\mathscr{L}_\xi \,g \right) \mathcal{A}^\beta+ g\, \mathscr{L}_\xi\mathcal{A}^\beta \nonumber \\
							&=\left(\mathscr{L}_\xi \,g \right) \mathcal{A}^\alpha\mathcal{A}^{\beta-\alpha} +g\, \mathscr{L}_\xi\left(\mathcal{A}^\alpha\mathcal{A}^{\beta-\alpha}\right) \nonumber \\
							&=- g\, \left(\mathscr{L}_\xi\mathcal{A}^\alpha\right) \mathcal{A}^{\beta-\alpha} +g\, \left(\mathscr{L}_\xi\mathcal{A}^\alpha\right)\mathcal{A}^{\beta-\alpha} \nonumber \\
							&\quad +g\, \mathcal{A}^\alpha \, \mathscr{L}_\xi\left(\mathcal{A}^{\beta-\alpha}\right) \nonumber \\
							&=h_\alpha \, \mathscr{L}_\xi\left(\mathcal{A}^{\beta-\alpha}\right).
\end{align}
The matrix power function is invertible, so \eqref{eq:step3} and Proposition \ref{prop:2} imply
\begin{align}
\label{eq:step4}
\mathscr{L}_\xi\left(\mathcal{A}^{\beta-\alpha}\right)=0\Longrightarrow \mathscr{L}_\xi\mathcal{A}=0,
\end{align}
which, thanks to the Lemma, implies
\begin{align}
\label{eq:step5}
\mathscr{L}_\xi\mathcal{A}=0\Longrightarrow \mathscr{L}_\xi \left(\mathcal{A}^\gamma\right)=0, \quad \forall \gamma \in \mathbb{R} .
\end{align}
Then, \eqref{eq:step2} tells us
\begin{align}
\label{eq:step6}
\left(\mathscr{L}_\xi \,g \right) \mathcal{A}^\alpha=0 \Longrightarrow \mathscr{L}_\xi \,g=0.
\end{align}
Therefore,
\begin{align}
\label{eq:thesis}
\mathscr{L}_\xi h_\gamma = \left(\mathscr{L}_\xi \,g \right) \mathcal{A}^\gamma+ g\, \mathscr{L}_\xi\left(\mathcal{A}^\gamma\right)=0,
\end{align}
$\forall \gamma \in \mathbb{R}$.
\end{proof}

\begin{cor*}
\label{prop:corollary}
  An isometry is bimetric if, and only if, it is an isometry for any two metrics in the set $\Gamma$.
\end{cor*}

Proposition \ref{prop:means} tells us that, if any two of the $h_\alpha $ share the KVFs, then every tensor field in the theory will be invariant under the corresponding transformation. The corollary tells us that this is a bimetric spacetime symmetry, and the whole spacetime is invariant under the transformation. In the opposite case, instead, the definition of spacetime symmetry needs more care. Note that, when stress--energy tensors are present, tensor fields within the same sector can have different collineations [compare with \eqref{eq:CollineationFE2} in which the collineations of $V_g{}^\mu{}_\nu$ and $T_g{}^\mu{}_\nu$ are not necessarily isometries of the metric $g_{\mu\nu}$; they are narrow symmetries]. However in vacuum, thanks to \eqref{eq:CollineationV}, tensors within the same sector share the collineations, which are then sectoral.

\subsection{Brief summary of \autoref{sec:HR}}
\label{subsec:summaryI}

Here we summarize the main results of this section.
\begin{enumerate}
	\item We presented example solutions of the bimetric field equations having different KVFs, different isometry groups or both; other methods for finding ansatzes with these properties are described in Appendix \ref{app:ansatz}.
	\item In HR bimetric theory, isometries arise in three configurations:
	\begin{enumerate}
	\item The two metrics share the KVFs, and the latter define symmetries for all other tensors.
	\item The two metrics have different KVFs.
	\begin{enumerate}
		\item In vacuum, tensors in the $g$-sector share the symmetries with $g_{\mu\nu}$, and those in the $f$-sector share them with $f_{\mu\nu}$.
		\item In the presence of stress--energy tensors, tensors in the $g$-sector need not share the symmetries with $g_{\mu\nu}$, and those in the $f$-sector need not share them with $f_{\mu\nu}$.
	\end{enumerate}
\end{enumerate}
	\item Case (a) can certainly be considered as an ``authentic" spacetime symmetry in bimetric relativity, i.e., a spacetime symmetry under which all the tensors in the theory are invariant. This is always the case in GR, since an isometry of \emph{the} metric is a collineation for the Riemann, Ricci, Einstein and stress--energy tensor. However, we also highlight the other possibilities in bimetric relativity. For example, consider two stationary metrics with different timelike KVFs in vacuum; in this case, it would be meaningful to consider this as a stationary spacetime, even if the KVFs are different. Hence, in general, one may define an authentic spacetime symmetry in HR bimetric theory referring to the minimal symmetry group shared by the metrics, and not to the associated symmetry vector fields.
	\item We proved that the two metrics share their KVFs if, and only if, any two of the metrics in the set $\Gamma=\left\lbrace h_\alpha = g\left(g^{-1}f \right)^\alpha : {\alpha \in \mathbb{R}}\right\rbrace$ have the same KVF.
	\item If $\eta^\mu$ and $\xi^\mu$ are KVFs of $g_{\mu\nu}$ and $f_{\mu\nu}$, respectively,
	then the following vector currents are covariantly conserved,
	$$ 
	   \nabla_\mu \left( V_g^{\mu\nu}\eta_\mu \right)\equiv 0, \quad
	   \tilde\nabla_\mu \left( V_f^{\mu\nu}\xi_\nu \right)\equiv 0,
	$$
	where $V_g^{\mu\nu}$ and $V_f^{\mu\nu}$ are the contributions from the bimetric potential to the field equations [see \eqref{eq:EoM}].
	This is analogous to the current determined by a stress--energy tensor,
	$$ 
       \nabla_\mu \left( T_g^{\mu\nu}\eta_\mu \right)\equiv 0,\quad
       \tilde\nabla_\mu \left( T_f^{\mu\nu}\xi_\nu \right)\equiv 0.
    $$
	In asymptotically flat spacetime, this leads to defining an energy density contribution from the tensor potentials $V_g^{\mu\nu}$ and $V_f^{\mu\nu}$, see \eqref{eq:conservedscalar}--\eqref{eq:energy}.
\end{enumerate}

\section{Geometric background and general analysis}
\label{sec:geometry}

In this section, we study spacetime symmetries on a differentiable manifold endowed with two generic metrics, without imposing any field equations. We only make geometrical considerations and generalize the analysis introduced in the first section. A short non-technical summary of the results is provided below, for the convenience of the reader not interested in the technicalities.

We want to understand if, starting with one metric having some isometries, it is possible to constrain the isometries of the second metric. The results are:

\begin{enumerate}
	\item Consider having two metrics on the same spacetime, and suppose one metric has some symmetries. Then, the second metric need not to share the symmetries with the first. Relations between the KVFs of the two metrics can always be found, but without additional requirements regarding the symmetries, these relations allow for any possible symmetry configuration.
	\item One possible requirement is that the metrics have KVFs with the same character (timelike, spacelike or null). For instance, a sufficient condition for the metrics $g_{\mu\nu}$ and $f_{\mu\nu}$ to have timelike KVFs $\xi^\mu$ and $\eta^\mu$ (i.e., to be stationary) is that they satisfy (in matrix notation) $g=P^\tr\,f\,P$ for some invertible $P$, and $\eta = P\, \xi$. The same is true if we require both metrics to be axially symmetric, and in general if we require KVFs with the same character.
	\item Another requirement is that the metrics share some isometries. This allows us to relate the KVFs of the metrics. They need not to be the same, but they must satisfy the same commutation relations; e.g., if we require spherical symmetry for both metrics, then the KVFs of both metrics must satisfy the angular momentum commutation relations
	\begin{align}
			\left[ J_i,J_j \right]=\epsilon ^{ijk}J_k.
		\end{align}
\end{enumerate}
We also add here the topological constraint of \autoref{subsec:HR}, because it is as general as the previous statements.
\begin{enumerate}[resume]
	\item A spacetime can only have one topology, and the two metrics must be compatible with it. This restricts the possible metric combinations, independently on the symmetries.
\end{enumerate}
These outcomes can be applied to any modified theory of gravity relying on (pseudo-)Riemannian geometry.

The rest of the section is quite technical; a discussion of the results can be found in the conclusions.

\subsection{The relation between torsion-free covariant derivatives}
\label{subsec:CovD}

On a differentiable manifold, given any two generic torsion-free covariant derivative operators $\nabla _\mu$ and $\tilde{\nabla}_\mu$, we can define the tensor ${C^{\alpha}}_{\mu\nu}$, symmetric in $\mu$ and $\nu$, such that:
\begin{equation}
\label{eq:DifferenceCovD}
\nabla _\mu \omega_\nu=\tilde{\nabla}_\mu\omega _\nu-C{}^{\alpha}{}_\mu{}_\nu\omega_\alpha,
\end{equation}
with $\omega _\alpha$ any 1-form defined on the cotangent space of the manifold 
\cite{Wald:1984rg}. Knowing \eqref{eq:DifferenceCovD}, one can 
straightforwardly deduce how the covariant derivatives of a tensor of any rank 
relate, and how the Riemann and Ricci tensors determined by the two covariant 
derivatives relate,
\begin{align}
\label{eq:DifferenceRiemann}
\tilde{R}{}_{\mu\nu\rho}{}^\sigma&=R{}_{\mu\nu\rho}{}^\sigma-2{C^{\sigma}}_{[\mu|\alpha}{C^{\alpha}}_{|\nu]\rho}+2\nabla_{[\mu|}{C^{\sigma}}_{|\nu]\rho}, \nonumber \\
\tilde{R}_{\mu\nu}&=R_{\mu\nu}-2{C^{\alpha}}_{[\mu|\beta}{C^{\beta}}_{|\alpha]\nu}+2\nabla_{[\mu|}{C^{\alpha}}_{|\alpha]\nu},
\end{align}
where $[\mu|...|\nu]$ denotes the antisymmetrization of $\mu$ and $\nu$ only. We notice that, in a coordinate basis, we can write
\begin{equation}
\label{eq:DifferenceChr}
C{}^{\alpha}{}_\mu{}_\nu=\Gamma{}^{\alpha}{}_\mu{}_\nu-\tilde{\Gamma}{}^{\alpha}{}_\mu{}_\nu,
\end{equation}
where the $\Gamma$'s are the Christoffel symbols of the covariant derivatives $\nabla _\mu$ and $\tilde{\nabla}_\mu$, respectively. If the two covariant derivative operators are compatible with two metrics $g_{\mu\nu}$ and $f_{\mu\nu}$, \eqref{eq:DifferenceChr} can be rewritten as,
\begin{align}
\label{eq:Cexplicit}
C{}^{\alpha}{}_\mu{}_\nu& =\dfrac{1}{2}g^{\alpha\beta}\left( \tilde{\nabla}_{\mu}g_{\beta\nu}+\tilde{\nabla}_{\nu}g_{\mu\beta}-\tilde{\nabla}_{\beta}g_{\mu\nu} \right) \nonumber \\
			& =-\dfrac{1}{2}f^{\alpha\beta}\left(\nabla_{\mu}f_{\beta\nu}+\nabla_{\nu}f_{\mu\beta}-\nabla_{\beta}f_{\mu\nu} \right).
\end{align}

\subsection{General properties of KVFs and the $A$-relation}
\label{subsec:Killing}

Consider an open set of a differentiable manifold. Suppose $\eta ^\mu$ and $\xi ^\mu$ are two vector fields defined on the spacetime and non-vanishing in the open set. Then, suppose we can find a linear local map between them having the following form,\footnote{The assumptions of non-vanishingness of the vector fields is not strictly necessary. We could allow for the vector fields to vanish the same number of times in the open set. However, if they vanish a different number of times, the linear local map between them will not exist (a linear map must send 0 to 0).}
\begin{align}
\label{eq:map}
\hxi^\mu(x^\alpha)	& \coloneqq\xi^\mu(\varphi^\beta\left(x^\alpha\right)) \nonumber \\
					& =\gamma^\mu{}_{\nu}\left(\varphi^\beta\left(x^\alpha\right)\right)\,\dfrac{\partial \varphi^{\nu}}{\partial x^{\rho}}(x^\alpha)\,\eta^{\rho}\left(x^\alpha\right) \nonumber \\
					& =A{}^\mu{}_\rho(x^\alpha)\,\eta^{\rho}(x^\alpha),
\end{align}
where we have defined
\begin{align}
\label{eq:defA}
A{}^\mu{}_\nu(x^\alpha)& \coloneqq \gamma^\mu{}_{\rho}\left(\varphi^\beta\left(x^\alpha\right)\right)\,\dfrac{\partial \varphi^{\rho}}{\partial x^{\nu}}(x^\alpha),
\end{align}
to increase readability.
In \eqref{eq:map}, we are applying two separate transformations to the spacetime and its tangent bundle. First, we apply the diffeomorphism $\varphi$ to the spacetime and we accordingly transform (push-forward) the vector field $\eta ^\mu$. Second, we apply a non-degenerate bundle map ${\gamma^\mu}_\nu$ over the spacetime (see e.g. \cite[p.~14]{husemoller1994fibre} and \cite[p.~116]{lee2003introduction}).\footnote{$A{}^\mu{}_\nu$ is a smooth automorphism of the tangent bundle onto itself.} Following the terminology in \cite[p.~87]{lee2003introduction}, we refer to $\xi^\mu$ and $\eta^\mu$ as ``$A$-related" KVFs, where $A$ refers to the composition of the two maps [see \eqref{eq:defA}].

We now assume $\eta ^\mu$ to be a KVF for $g_{\mu\nu}$ and write the Killing equation,
\begin{equation}
\label{eq:Lieg}
\mathscr{L}_{\eta}g_{\mu\nu}=2\nabla_{(\mu}\eta_{\nu)}=0,
\end{equation}
with $\mathscr{L}_{\eta}$ being the Lie derivative along $\eta^\mu$, $\eta _\mu\coloneqq g_{\mu\nu}\eta ^\nu$, and $\nabla_\mu$ being the compatible covariant derivative operator for $g_{\mu\nu}$. We require $\xi^\mu$ to be a KVF for $f_{\mu \nu}$,
\begin{equation}
\label{eq:Lief}
\mathscr{L}_{\xi}f_{\mu \nu}=2\tilde{\nabla}_{(\mu}\xi_{\nu)}=0,
\end{equation}
with $\xi _\mu\coloneqq f_{\mu \nu}\xi ^\nu$ and $\tilde{\nabla}_\mu$ being the compatible 
covariant derivative operator for $f_{\mu \nu}$.\footnote{The requirement of non-vanishingness for $\eta^\mu$ and $\xi^\mu$ is also motivated by their utilization as KVFs of the metrics. If a KVF vanishes at some point, a more detailed treatment should be carried out (see \cite[Section 10.4]{hall2004symmetries}).}

Using \eqref{eq:DifferenceCovD} and \eqref{eq:map} in \eqref{eq:Lief}, we can rewrite the Killing equation for $f_{\mu\nu}$ as follows,
\begin{align}
\label{eq:LieFinal}
\mathscr{L}_{\xi}f_{\mu \nu} 	& =2\tilde{\nabla}_{(\mu}\xi_{\nu)}\nonumber \\
								& =2\left(\nabla_{(\mu}\xi_{\nu)}+{C^\alpha}_{\mu \nu}\xi_\alpha\right) \nonumber \\
								& =2\left[\nabla_{(\mu}\left(f_{\nu)\rho}{A^\rho}_\sigma g^{\sigma\delta}\heta_\delta \right)+{C^\alpha}_{\mu \nu}f_{\alpha \rho}{A^\rho}_\sigma g^{\sigma\delta}\heta_\delta \right] \nonumber \\
								& =2\left[ \nabla_{(\mu}\left(f_{\nu)\rho}{A^\rho}_\sigma g^{\sigma\delta} \right)+{C^\alpha}_{\mu \nu}f_{\alpha \rho}{A^\rho}_\sigma g^{\sigma\delta}\right] \heta_\delta \nonumber \\
								&\quad +2\, f_{(\nu|\rho}{A^\rho}_\sigma g^{\sigma\delta}\left(\nabla_{|\mu)}\heta_\delta \right).
\end{align}
Defining the tensor,\footnote{The tensor $\Upsilon_\nu{}^\delta$ defined in \eqref{eq:adjoint} is the adjoint map of ${A^\mu}_\sigma$. Specifically, $\xi_\nu=\Upsilon_\nu{}^\delta\eta_\delta$.                                                                                                                                                                                                                                                                                                                                                                                                                                                                                                                                                                                                                                                                                                                                                                                                                                                   }
\begin{equation}
\label{eq:adjoint}
\Upsilon_\nu{}^\delta \coloneqq f_{\nu \rho}{A^\rho}_\sigma g^{\sigma\delta},
\end{equation}
\eqref{eq:LieFinal} becomes,
\begin{align}
\label{eq:LieSimple}
\mathscr{L}_{\xi}f_{\mu \nu} & =2\left[ \nabla_{(\mu}{\Upsilon_{\nu)}}^\delta+{C^\alpha}_{\mu \nu}{\Upsilon_\alpha}^\delta\right] \eta_\delta \nonumber\\
						&\quad +2\, \Upsilon {}_{(\nu|} {}^\delta \left(\nabla_{|\mu)}\eta_\delta \right)=0.
\end{align}
Note that we have not used the hypothesis $\nabla_{(\mu}\eta_{\nu)}=0$, because 
it never appears explicitly in \eqref{eq:LieSimple}. Therefore, we can 
deduce that the isometries of $g_{\mu\nu}$ and $f_{\mu\nu}$, in general, are unrelated, as one might expect.

Now, assuming \eqref{eq:KVFs} defines the KVFs of $g_{\mu\nu}$, we can 
find the general relation between them and their \arel KVFs of $f_{\mu\nu}$ by 
substituting \eqref{eq:DifferenceCovD} and \eqref{eq:DifferenceRiemann} in the analog of \eqref{eq:KVFs} for 
$f_{\mu\nu}$. However, we have not been able to deduce any information from the 
resulting equations other than that, in the generic case, the KVFs of two 
metrics are unrelated, as already shown in \eqref{eq:LieFinal}. Nonetheless, we have seen how to make use of \eqref{eq:KVFsb} in the case of HR bimetric theory.

An interesting viewpoint is to consider the two metrics related by a ``generalized" vielbein. In such an approach, the generalized vielbein $V {}^\mu{}_\nu$ would transform one metric into the other, e.g.
\begin{equation}
f_{\mu\nu}=V {}_\mu{}^\rho g_{\rho\sigma} V {}^\sigma{}_\nu,
\end{equation}
as the standard vielbein does when one of the metric is Minkowski. Then, a similar treatment as the one performed in \cite{Chinea:1988} is possible. There, it is shown that the standard vielbein must satisfy some specific constraints in order for the obtained metric to have a given set of isometries. In principle, as a standard vielbein connect the Minkowski metric with a completely different and generic metric, the same holds in our set-up for the generalized vielbein relating $g_{\mu\nu}$ and $f_{\mu\nu}$. Therefore, there cannot be a general relation between the isometries of the two metrics.\footnote{The ``generalized vielbein" is actually determined by the map in \eqref{eq:map}, which is not only a map between vector fields, but actually between the two whole sectors. Indeed, we have already introduced the adjoint map $\Upsilon {}_\nu{}^\delta$ relating the 1-forms; then, every tensor can be mapped from a sector into the other.}

\subsection{Constraints on the \arel KVFs and the isometry groups}
\label{subsec:characters}

In this subsection we find the constraints on the map $A{}^\mu{}_\nu$ in \eqref{eq:map} by requiring the KVFs of the two metrics to share certain properties. In particular, the properties will concern (i) the character of the KVFs and (ii) the Lie algebra that they generate (we refer the reader to \cite{choquet1982analysis,hall2003lie} for a rigorous treatment of the Lie groups and Lie algebras on a (pseudo-)Riemannian manifold).

Let's consider case (i) first. We explore the relation between the character (timelike, spacelike, or null) of two \arel KVFs and how it depends on the map $A{}^\mu{}_\nu$ in \eqref{eq:map}. Specifically, we discuss one condition that $A{}^\mu{}_\nu$ can satisfy in order for the \arel KVFs to have the same character with respect to their associated metric.

This issue is of physical interest because it concerns whether two metrics share or not some KVFs, both being timelike, spacelike or null with respect to the pertinent metric. For example, if the two metrics have different timelike \arel KVFs, then they are both stationary and could be both static (we will come back to this case later), but have different KVFs.

$\xi^\mu$ and $\eta^\mu$ having the same character with respect to $f_{\mu\nu}$ and $g_{\mu\nu}$ means
\begin{equation}
\label{eq:sameCharacter1}
\operatorname{sign}[g_{\alpha\beta}\eta^\alpha\eta^\beta]=\operatorname{sign}[f_{\alpha\beta}\xi ^\alpha\xi^\beta].
\end{equation}
One case in which this relation is certainly satisfied is when (in matrix notation)
\begin{equation}
\label{eq:congruence}
g=A^\tr f A,
\end{equation}
i.e., $A$ is a congruence between $g$ and $f$. Stated differently, the map $A$ between \arel KVFs is a local isometry between $g$ and $f$ (i.e., it is an isometry locally at each point of the spacetime). Then,
\begin{align}
g_{\alpha\beta}\eta^\alpha\eta^\beta=f_{\alpha\beta}\xi ^\alpha\xi^\beta,
\end{align}
which guarantees that $\xi ^\mu$ and $\eta ^\mu$ always have the same character. However, we remark that this is a sufficient, but not necessary, condition to satisfy \eqref{eq:sameCharacter1}.

Next, we consider case (ii). We address the question of when two metrics share their isometry groups, i.e., they determine the same spacetime symmetry. Even if the two metrics have KVFs of the same character, we are not guaranteed that the metrics have the same isometry group. The simplest case concerns stationarity, because it only requires a timelike KVF \cite{Wald:1984rg}. Suppose that one metric, for example $g_{\mu\nu}$, possess a timelike KVF $\eta ^\mu$. Then, if $A{}^\mu{}_\nu$ is a congruence between the metrics $g_{\mu\nu}$ and $f_{\mu\nu}$, the \arel KVF of $\eta^\mu$, i.e., $\xi^\mu$, will be timelike with respect to $f_{\mu\nu}$. Therefore, both metrics will be stationary.\footnote{This result relies on the fact that all 1-dimensional Lie algebras are isomorphic, because the Lie bracket is trivially zero.} However, this does not take into account staticity. Each of the metrics can be static (thus having a different isometry group), if, and only if, their timelike KVFs are orthogonal to a congruence of spacelike hypersurfaces \cite{Wald:1984rg}. No general relation between the \arel KVFs is manifest in this case. Therefore, one metric could be only stationary and the other static; i.e., different isometry groups, where the isometry group of a static metric is a subgroup of the stationary isometry group, the latter including the time reversal.

In the most general case, i.e., when considering arbitrary isometry groups, the relation between the KVFs of two metrics can be analyzed in terms of the Lie algebras of the isometry groups (which are Lie groups) of the two metrics. Let $\{ \eta _{(n)}^\mu \}_{n\in \{1,...,\delta\}}$, with $\delta$ dimension of the isometry group, be a subset of the KVFs of $g_{\mu\nu}$ generating a Lie algebra defined by the Lie bracket,
\begin{align}
\label{eq:LieBracket}
\left[ \eta _{(n)},\eta _{(m)} \right](F)\coloneqq\eta _{(n)}\left( \eta _{(m)}(F) \right)-\eta _{(m)}\left( \eta _{(n)}(F) \right),
\end{align}
where $m\in \{1,...,\delta\}$ and $F$ is a $C^\infty $ scalar function defined on the differentiable manifold. In order $f_{\mu\nu}$ to have the same isometry group as $g_{\mu\nu}$, a necessary condition is that the set of \arel KVFs $\{ A(\heta _{(n)}^\mu) \}_{n\in \{1,...,\delta\}}$ generates a Lie algebra which is isomorphic to the one generated by the KVFs of $g_{\mu\nu}$. This implies that the map $A{}^\mu{}_\nu$ between \arel KVFs is a Lie algebra isomorphism. We remark that, even in the case of \arel KVFs generating the same Lie algebra, they need not to be the same. Indeed, a Lie algebra can have different generators (i.e., different bases), and different choices determine different KVFs, as we saw explicitly in the examples in \autoref{subsec:explicit}. Also, we recall that different Lie groups can have the same Lie algebra, and that is why $A{}^\mu{}_\nu$ being a Lie algebra isomorphism is a necessary condition, but not sufficient to have the same isometry group.\footnote{However, Lie groups sharing the Lie algebra have the same universal covering \cite{hall2003lie}.}

\subsection{Two particular cases}
\label{subsec:particularCases}

In this subsection we discuss two relevant cases allowing us to understand more deeply our examples in \autoref{subsec:explicit}.

The first case, concerning an isomorphism between Lie algebras, is when $\gamma{}^\mu{}_\nu=\delta{}^\mu{}_\nu$ in \eqref{eq:map} and the map $A{}^\mu{}_\nu =\left(\varphi _*\right){}^\mu{}_\nu=(\dd \varphi){}^\mu{}_\nu$ is simply the differential map (push-forward) of a diffeomorphism $\varphi$ from the spacetime to itself.\footnote{One can argue that this 
must be the case, because (active) diffeomorphism of a manifold into 
itself can be thought as (passive) change of coordinates, which cannot 
modify the geometric properties of the manifold.} Indeed, the push-forward is a linear map that preserves the Lie bracket, 
(see Proposition 4.3.10 in \cite{Conlon:2001diff}), i.e., for any two vector
fields $X$ and $Y$
\begin{align}
\label{eq:diffBracket}
\varphi_* \left[ X,Y \right] = \left[ \varphi_*X,\varphi_*Y \right],
\end{align}
i.e., $\varphi$-related KVFs. This case concerns metrics which have the same form in different coordinate systems (i.e., related by a diffeomorphism). At first 
glance, this case may seem trivial, but actually it is not, since the 
\emph{same} coordinate system can be, e.g. spherical for one metric and not for 
the other. The second example and some cases within the the first example in \autoref{subsec:explicit}, as discussed there, have $\varphi$-related KVFs, hence belonging to this category.

The second case we consider is when no diffeomorphism in \eqref{eq:map} is applied and
\begin{equation}
A{}^\mu{}_\nu=\gamma{}^\mu{}_\nu=f^{\mu\rho}g_{\rho\nu}.
\end{equation}
An example belonging to this category will be treated in Appendix \ref{app:ansatz}. In this case, $\Upsilon{}_\mu{}^\nu=\delta{}_\mu{}^\nu$ and the two vector fields $\xi ^\mu$ and $\eta ^\mu$ have the same dual,
\begin{equation}
\xi _\mu=f_{\mu\rho}\xi ^\rho=g_{\mu\rho}\eta^\rho=\eta _\mu.
\end{equation}
In addition, \eqref{eq:LieSimple} reduces to,
\begin{align}
\label{eq:sameCovectors}
\mathscr{L}_{\xi}f_{\mu\nu} & =\left[ 2\nabla_{(\mu}{\delta_{\nu)}}^\delta+{C^\alpha}_{\mu\nu}{\delta_\alpha}^\delta\right] \eta_\delta +2\, {\delta_{(\nu|}}^\delta \left(\nabla_{|\mu)}\eta_\delta \right) \nonumber \\
						& ={C^\alpha}_{\mu\nu} \eta_\alpha +2\nabla_{(\mu}\eta_{\nu)}={C^\alpha}_{\mu\nu} \eta_\alpha \nonumber \\
						&=\left({\tilde{\Gamma}^\alpha}_{\mu\nu}-{\Gamma^\alpha}_{\mu\nu}\right) \eta_\alpha=0,
\end{align}
where, in the second line, we have used our original hypothesis in \eqref{eq:Lieg}, i.e., 
$\nabla_{(\mu}\eta_{\nu)}=0$. With \eqref{eq:sameCovectors} we show that, for $A{}^\mu{}_\nu=f^{\mu\rho}g_{\rho\nu}$, $\xi ^\mu$ is a KVF for $f$ if, and only if, $\eta _\mu$ is orthogonal to ${C^\alpha}_{\mu\nu}$.\footnote{Equivalently, \eqref{eq:sameCovectors} is an eigenequation for $C^\alpha{}_{\mu\nu}$ and $\eta^\mu$ must be an eigenvector of $C^\alpha{}_{\mu\nu}$ with eigenvalue zero.} One particular solution of \eqref{eq:sameCovectors} is obtained when $g_{\mu\nu}$ and $f_{\mu\nu}$ share their compatible covariant derivatives,
\begin{align}
\label{eq:sameCovD}
\nabla_\mu=\tilde{\nabla_\mu} \Longrightarrow {\Gamma^\alpha}_{\mu\nu}=\tilde{\Gamma}^\alpha_{\mu\nu}.
\end{align}
It is well known that, given a metric on a differentiable manifold, its compatible covariant derivative is uniquely determined (see e.g. \cite[Theorem 3.1.1]{Wald:1984rg}). However, the converse is not true, as was shown in \cite{Hall:1983a,Hall:1983b,hall2004symmetries}. Given a covariant derivative operator on a differentiable manifold, there can be more than one metric compatible with it. There are four possible cases, one of them being when the two metrics are proportional (the most commonly encountered, according to \cite{Hall:1983a,Hall:1983b}). Arguably, this is the least interesting case in physics, because, the metrics being proportional, their null cones (hence their causal structures) are the same.

\section*{Conclusions}

We study spacetime symmetries and topology in the HR bimetric theory, where two metrics $g_{\mu\nu}$ and $f_{\mu\nu}$ are defined on the same manifold. We determine the conditions for the metrics to share their isometries, but we note they can also have different isometries. In detail, we present a proposition stating that, if any two metrics within the set $\left\lbrace h_\alpha \coloneqq g\left(g^{-1}f \right)^\alpha : {\alpha \in \mathbb{R}}\right\rbrace$ have the same KVF, then all other fields in the theory must have the same KVF. Also, we find a differential equation which determines a KVF shared by both metrics.

We point out that many properties of spacetime symmetries, valid in GR, are not valid in HR bimetric theory. For instance, an isometry of one metric is not necessarily a collineation for the minimally coupled stress--energy tensor, due to the tensor potentials $V_g{}^\mu{}_\nu$ and $V_f{}^\mu{}_\nu$ in the bimetric field equations. Also, in vacuum, a KVF $\xi ^\mu$ does not have to satisfy $\dalamb\xi^\mu=0$. However, we pointed out another geometrical property of a KVF, following from the Komar identity, i.e., $\nabla _{\mu}\left(\dalamb\xi^\mu \right)=0$, for which we provide an alternative proof in appendix \ref{app:alternativeproof}.

In HR bimetric theory, concerning collineations of the tensors in the theory, three configurations are possible:
\begin{enumerate}
	\item The two metrics have the same KVFs, which define collineations for every tensor in both sectors.
	\item The two metrics do not have the same KVFs.
	\begin{enumerate}
		\item In vacuum, tensors in the $g$-sector have the same collineations as $g_{\mu\nu}$, and those in the $f$-sector the same as $f_{\mu\nu}$.
		\item When stress--energy tensors are present, tensors in the $g$-sector do not necessarily have the same collineations as $g_{\mu\nu}$, and tensors in the $f$-sector do not necessarily have the same collineations as $f_{\mu\nu}$.
	\end{enumerate}
\end{enumerate}

Therefore, contrary to GR, an isometry of the metric is not necessarily a spacetime symmetry for all other tensors. This raises the intriguing question of what the best definition of a spacetime symmetry in HR bimetric theory is. If the two metrics share the isometry, then all tensors have the same spacetime symmetry, similar to GR. In this case, the isometries of the metrics are spacetime symmetries for the whole system and we call them ``bimetric symmetries." It seems reasonable to consider bimetric symmetries as ``authentic" spacetime symmetries in bimetric relativity---i.e., symmetries for every tensors in the theory---but we stress that there are other possibilities that can be considered as authentic. For instance, we can have symmetries of one sector only, which we call ``sectoral symmetries", and symmetries of some tensors only, named ``narrow symmetries". Sectoral symmetries can be regarded as authentic, in some cases. For example, we can have two spherically symmetric metrics with different KVFs in vacuum (as we show in Example III in \autoref{subsec:explicit}). One can think about this as a spherically symmetric spacetime. Therefore, it seems reasonable to define an authentic spacetime symmetries as the minimal symmetry group shared by the sectors, without reference to the associated symmetry vector fields.

We clarify that a topological constraint limits the number of conceivable metric combinations. The unique diffeomorphism invariance of the theory allows us to have only a single set of coordinate charts covering the spacetime. Such a set is compatible with one, and only one topology, and both metrics must be compatible with it.

We present examples showing the relations between the spacetime symmetries, the KVFs and the topologies of the two sectors. We provide several methods to determine possible ansatzes with metrics having different isometry groups, different KVFs or both in Appendix \ref{app:ansatz}.

Interesting open questions remain unanswered in our analysis. First, it would be interesting to find non-GR (analytical or numerical) solutions not sharing their KVFs and/or their isometry groups. Second, our study does not involve the fixed point structure of the KVFs' flow, and it would be desirable to understand the relations between these structures in the two sectors. Third, we have not considered the relationship between KVFs of the two metrics and Lie point symmetries  \cite[p.~129]{stephani_kramer_maccallum_hoenselaers_herlt_2003} of the bimetric field equations.

In principle, new ansatzes can be found and used in HR bimetric theory if we allow for different KVFs, or different isometries of the two metrics, a fact that has gained little attention in the literature so far. Our results clarify that this is definitely a promising possibility to enlarge the spectrum of (hopefully exact) solutions in bimetric relativity.

\begin{acknowledgments}

We are grateful to Ingemar Bengtsson for many rewarding discussions, for reading the manuscript and providing valuable comments. We also thank Fawad Hassan, Bo Sundborg, Kjell Rosquist, Luis Apolo, Nico Wintergerst and Mikael von Strauss for many helpful and inspiring discussions.

\end{acknowledgments}

\appendix

\section{How to determine ansatzes having different KVFs}
\label{app:ansatz}

In the main text, we have studied some particular solutions of the bimetric field equations, focusing on their isometries and on the shared spacetime topology. We concluded that there are both solutions where the metrics share their isometry group but not their KVFs, and solutions where the metrics do not have the same isometry group. In this appendix we devise three different methods of generating infinite sets of solutions belonging to these two categories. Hence the particular examples studied in the main text are not unique.

In the first example, the two metrics will be Einstein metrics with vanishing cosmological constants and vanishing stress--energy tensors in both sectors, i.e.
\begin{subequations}
\label{eq:Einstein}
\begin{align}
G_g{}^\mu{}_\nu = G_f{}^\mu{}_\nu &= 0, \\
T_g{}^\mu{}_\nu = T_f{}^\mu{}_\nu &= 0.
\end{align}
\end{subequations}
The field equations to be solved are
\begin{align}
\label{eq:tensorpotentials}
V_g{}^\mu{}_\nu = V_f{}^\mu{}_\nu = 0.
\end{align}
We note that, if the field equations \eqref{eq:EoM} hold, the Bianchi constraints \eqref{eq:Bianchi} are automatically satisfied.\footnote{If we solve the bimetric field equations \eqref{eq:EoM} without using the Bianchi constraints \eqref{eq:Bianchi}, the solution will satisfy the Bianchi constraints automatically. The latter are enclosed and hidden in the field equations. However, the explicit use of them makes it simpler to solve the field equations.}
In the second example, we introduce stress--energy tensors in both sectors to cancel the respective Einstein tensors and again we are left to solve \eqref{eq:tensorpotentials}. In the last example, we consider a non-GR solution.

\subsection*{Method 1: Generating Minkowski solutions with Lorentz transformations}
We pick by hand the metric $g_{\mu\nu}$ to be the Minkowski metric. Then there exists a coordinate chart $x^\mu=(t,x,y,z)$ such that the components of $g_{\mu\nu}$ are
\begin{align}
\label{eq:gMink}
	g_{\mu\nu} = \eta_{\mu\nu} = \mathrm{diag}(-1,1,1,1).
\end{align}
Note that there are ten degrees of freedom (DOFs) in each of the metrics and hence twenty in total. Choosing $g_{\mu\nu}$ as in \eqref{eq:gMink}, its DOFs are completely determined. Since $f=gS^2$, the ten undetermined DOFs in $f$ can be redistributed to the square root matrix $S$. For Einstein metrics, the equations \eqref{eq:tensorpotentials} completely determine the four eigenvalues $\lambda_1,...,\lambda_4$ of the square root matrix in terms of the $\beta$-parameters and one of $\beta$-parameters in terms of the others \cite{Kocic:2017wwf}.\footnote{Unless for very a special set of the $\beta$-parameters, referred to as algebraically decoupled in \cite{Kocic:2017wwf}.} This leaves six undetermined DOFs in $S{}^\mu{}_\nu$, which can be determined by choosing its diagonalizing matrix. Here we choose a constant Lorentz transformation $\Lambda{}^\mu{}_\nu$ so that 
\begin{align}
	(\Lambda^{-1}){}^\mu{}_\rho S{}^\rho{}_\sigma \Lambda{}^\sigma{}_\mu = \mathrm{diag}(\lambda_1,\lambda_2,\lambda_3,\lambda_4).
\end{align}
With this choice, $f_{\mu\nu}$ simply becomes a Minkowski metric,
\begin{align}
\label{eq:f_ansatz}
	f_{\mu\nu} = (\Lambda^{-1,\tr}){}_\mu{}^\rho \mathrm{diag}(-\lambda_1^2,\lambda_2^2,\lambda_3^2,\lambda_4^2)_{\rho\sigma} (\Lambda^{-1}){}^\sigma{}_\nu.
\end{align}
It is easy to show that $f_{\mu\nu}$ is a Minkowski metric, albeit that the Cartesian coordinates of $f_{\mu\nu}$ are different from those of $g_{\mu\nu}$. Defining a new set of coordinates $x'^\mu$ via
\begin{align}
\label{eq:diff}
	\frac{\partial x'^\mu}{\partial x^\nu} = \mathrm{diag}(\lambda_1,\lambda_2,\lambda_3,\lambda_4)\Lambda^{-1},
\end{align}
this is a set of first-order linear partial differential equations and always has a solution for $x'^\mu$ as an (invertible) function of $x^\mu$, so it is a valid coordinate transformation and $f_{\mu'\nu'} = \mathrm{diag}(-1,1,1,1)$. Hence $0=G_f{}^{\mu}{}_{\nu} = G_f{}^{\mu'}{}_{\nu'}$. We conclude that 
\begin{align}
g_{\mu\nu} &= -\mathrm{d}t^2 + \mathrm{d}x^2 + \mathrm{d}y^2 + \mathrm{d}z^2, \nonumber\\
f_{\mu\nu} &= \left[{{(\Lambda^{-1,\tr})}_\mu}^\rho \mathrm{diag}(-\lambda_1^2,\lambda_2^2,\lambda_3^2,\lambda_4^2)_{\rho\sigma} {{(\Lambda^{-1})}^\sigma}_\nu\right]\mathrm{d}x^\mu \mathrm{d}x^\nu \nonumber\\
&= -\mathrm{d}t'^2 + \mathrm{d}x'^2 + \mathrm{d}y'^2 + \mathrm{d}z'^2,
\end{align}
solves the bimetric field equations. Note that $\Lambda{}^\mu{}_\nu$ is an arbitrary constant Lorentz matrix. Choosing different Lorentz matrices (i.e., choosing different Lorentz diffeomorphisms between the sectors) yields different solutions. Since $g$ and $f$ are related via the diffeomorphism \eqref{eq:diff} they must have the same symmetry group, as explained in \autoref{subsec:particularCases}. However, the Cartesian charts of $g_{\mu\nu}$ and $f_{\mu\nu}$ do not coincide generally and hence the KVFs of $g_{\mu\nu}$ and $f_{\mu\nu}$ may differ; they are related by the Lorentz transformation $\Lambda{}^\mu{}_\nu$.

As an example, let $\beta_{0,2,3}=1$, $\beta_1 = 0$, $\beta_4 = 2$, $\lambda_{1,3}=(1+\sqrt{5})/2$, and $\lambda_{2,4}=(1-\sqrt{5})/2$. This choice of parameters satisfies $V_g{}^\mu{}_\nu = V_f{}^\mu{}_\nu = 0$. We choose the Lorentz matrix to be a boost of velocity 1/2 in the $x$-direction combined with a rotation of angle $-\pi/4$ around the $x$-axis. $f_{\mu\nu}$ becomes
\begin{align}
\label{eq:fsol}
f_{\mu\nu}&=
\begin{pmatrix}
-\frac{1}{6}\left(9+5\sqrt{5}\right) & -\frac{2\sqrt{5}}{3} & 0 & 0 \\
-\frac{2\sqrt{5}}{3} & \frac{1}{6}\left(9-5\sqrt{5}\right) & 0 & 0 \\
0 & 0 & \frac{3}{2} & -\frac{\sqrt{5}}{2} \\
0 & 0 & -\frac{\sqrt{5}}{2} & \frac{3}{2} \\
\end{pmatrix}\!.
\end{align}
It is now straightforward to check whether or not the KVFs of $g_{\mu\nu}$ and $f_{\mu\nu}$ coincide. The KVF generating the rotational symmetry of $g_{\mu\nu}$ around the $x$-axis is $\eta=-z\partial_y+y\partial_z$. The Lie derivative of $f_{\mu\nu}$ with respect to $\eta^\mu$ does not vanish. Rather, the KVF of $f_{\mu\nu}$ generating the rotational symmetry around the $x'$-axis is $\xi = -z'\partial_{y'}+y'\partial_{z'}$, or, expressed in the Cartesian chart $x^\mu$ for $g_{\mu\nu}$,
\begin{align}
	\xi = \left(-\sqrt{5}y+3z\right)\partial_y+\left(-3y+\sqrt{5}z\right)\partial_z.
\end{align}
Then, the metrics share their isometry group, but not all of their KVFs.  

Note how the assumption of a constant $\Lambda{}^\mu{}_\nu$ in \eqref{eq:f_ansatz} insured $g_{\mu\nu}$ and $f_{\mu\nu}$ to be related via a diffeomorphism and hence $G_f{}^\mu{}_\nu$ to vanish. The method of generating solutions could be generalized by retaining the spacetime dependence of the Lorentz matrix. In that case however, $G_f{}^\mu{}_\nu=0$ does not hold automatically and becomes a differential equation for $\Lambda(x){}^\mu{}_\nu$.

When solving the bimetric field equations it is common to demand that a flat spacetime solution exists in the theory. Usually this is achieved by demanding $\eta_{\mu\nu} = g_{\mu\nu} = c^2 f_{\mu\nu}$ to be a solution of the field equations, yielding the conditions, sometimes called asymptotic flatness conditions,
\begin{align}\label{eq:sc-assy}
\beta_0 &= -3\beta_1-3\beta_2 - \beta_3, \nonumber \\
\beta_4 &= -\beta_1 - 3\beta_2 -3\beta_3.
\end{align}
If we are minimally coupled to $g_{\mu\nu}$, the trajectory of a point-like test particle will be a solution of the geodesic equation defined by $g_{\mu\nu}$. Hence all measurable geometrical quantities will be contained in the metric $g_{\mu\nu}$. Accordingly, demanding $g_{\mu\nu}=\eta_{\mu\nu}$ is sufficient in order to obtain a universe that we would measure as flat. In fact, even in the case where both metrics are Minkowski, the solutions \eqref{eq:gMink}, \eqref{eq:fsol} shows that the $\beta$-parameters does not need to satisfy the conditions \eqref{eq:sc-assy}. Therefore, demanding \eqref{eq:sc-assy} is unnecessarily restrictive in order for a bi-flat solution to exist in the theory. On the other hand, demanding $g_{\mu\nu}=\eta_{\mu\nu}$ does constraint $f_{\mu\nu}$ to have a specific form \cite{Kocic:2017wwf}.

\subsection*{Method 2: Decoupled interaction terms}

Let's choose a Lorentzian metric $g_{\mu\nu}$ and use Synge's method to define the stress--energy tensor minimally coupled to $g_{\mu\nu}$. Then, $T_g{}^\mu{}_\nu$ is proportional to the Einstein tensor of $g_{\mu\nu}$, i.e.
\begin{align}
\label{eq:TgGg}
	T_g{}^\mu{}_\nu := M_g^2\, G_g{}^\mu{}_\nu.
\end{align}
For this construction to yield a physical solution it is important to pick the metric $g_{\mu\nu}$ yielding the stress--energy tensor in \eqref{eq:TgGg} satisfying all the desired properties, e.g. the null energy condition (NEC). For example, if $g_{\mu\nu}$ is the FLRW metric for a homogeneous and isotropic matter dominated universe, $g_{\mu\nu}=-\mathrm{d}t^2+t^{4/3}\mathrm{d}\vec{x}$, then the stress--energy as defined by \eqref{eq:TgGg} is physical, since we know that its Einstein tensor is proportional to $\mathrm{diag}(t^{-2},0,0,0)$ and so it describes a pressure-less perfect fluid. In such a case, the tensor potential $V_g{}^\mu{}_\nu$ decouples from the field equations and we are left with $V_g{}^\mu{}_\nu=0$.

Now let $S{}^\mu{}_\nu$ be a rank (1,1) tensor such that $gS=(gS)^\tr$ and $V_f{}^\mu{}_\nu(S)=0$. The first equation tells us that $h:=gS$ is a symmetric rank (0,2) tensor, implying that the $f$ metric, $f:=hS$, is also symmetric. To be explicit, $f^\tr := S^\tr h^\tr = S^\tr h = S^\tr g S=(gS)^\tr S=f$. Now, defining $T_f{}^\mu{}_\nu$ to be $T_f{}^\mu{}_\nu := M_f^2\, G_f{}^\mu{}_\nu$, the bimetric field equations are satisfied. Concerning the Bianchi constraints, they are satisfied by construction since both tensor potentials vanishes and the stress--energy tensors in both sectors are defined to be the Einstein tensors of the two metrics.

Note that the only equations that needs to be satisfied when implementing the method are $gS=(gS)^\tr$ to obtain a symmetric $f_{\mu\nu}$ metric and $V_{g,f}{}^\mu{}_\nu(S)=0$. The solution of the last two equations determines in the general case the (constant) eigenvalues of $S{}^\mu{}_\nu$ in terms of the $\beta$-parameters and determines one of the $\beta$-parameters in terms of the others. 

To illustrate the method, let $g_{\mu\nu}$ be the Minkowski metric so that the ten DOFs in $g_{\mu\nu}$ are completely specified. Then the Einstein tensor for $g_{\mu\nu}$ vanishes and accordingly $T_g{}^\mu{}_\nu:=0$. Furthermore, assume that
\begin{align}
\label{eq:STypeI}
(\Lambda^{-1}(x)){}^\mu{}_\rho S{}^\rho{}_\sigma \Lambda(x){}^\sigma{}_\nu = \mathrm{diag}(\lambda_1,\lambda_2,\lambda_3,\lambda_4),
\end{align}
where $\Lambda(x)$ is a non-constant, local Lorentz matrix. Now, the ten remaining DOFs are distributed as four in the eigenvalues of $S{}^\mu{}_\nu$ and six in the Lorentz matrix. From \eqref{eq:STypeI} we see that $gS=(gS)^\tr$ is satisfied by construction. To satisfy $V_{g,f}{}^\mu{}_\nu(S)=0$ we solve for the eigenvalues of $S{}^\mu{}_\nu$ and for $\beta_4$, so that $\lambda_i=\lambda_i(\beta_0,...,\beta_4)$ and $\beta_4=\beta_4(\beta_0,...,\beta_3)$. For definiteness, let ${{\Lambda(x)}}$ be a rotation around the $z$-axis with the rotation angle depending on $z$,
\begin{align}
\Lambda^{-1,\tr}(z)&=
\begin{pmatrix}
1 & 0 & 0 & 0 \\
0 & \cos (z) & - \sin (z) & 0 \\
0 & \sin (z) & \cos (z) & 0 \\
0 & 0 & 0 & 1 \\
\end{pmatrix}\!.
\end{align}
$f_{\mu\nu}$ becomes
\begin{align}
f_{\mu\nu}&= \nonumber \\
&\hspace{-6mm}\begin{pmatrix}
-\lambda_1^2 & 0 & 0 & 0 \\
0 & \lambda_2^2 \cos^2 (z) + \lambda_3^2  \sin^2 (z) & \frac{1}{2}\left(\lambda_3^2-\lambda_2^2\right)\sin (2z) & 0 \\
0 & \frac{1}{2}\left(\lambda_3^2-\lambda_2^2\right)\sin (2z) & \lambda_3^2 \cos^2 (z) + \lambda_2^2 \sin^2 (z) & 0 \\
0 & 0 & 0 & \lambda_4^2 \\
\end{pmatrix}\!.
\end{align}
The Einstein tensor for $f_{\mu\nu}$ can then be computed and we define $T_f{}^\mu{}_\nu := M_f^2\, G_f{}^\mu{}_\nu$. The explicit expression is not presented here but it is important to note that if $\lambda_2^2 \neq \lambda_3^2$, it is not vanishing. Thus, $f_{\mu\nu}$ is not the Minkowski metric. Interestingly enough, the Ricci scalar of $f_{\mu\nu}$ is a constant depending on the square root eigenvalues $\lambda_i$
\begin{align}
	R_f = -\frac{(\lambda_2^2-\lambda_3^2)^2}{2\lambda_2^2 \lambda_3^2 \lambda_4^2}.
\end{align}
However, the Ricci tensor of $f_{\mu\nu}$ is not proportional to the identity matrix,
\begin{align}
R_f{}^\mu{}_\nu &\propto \nonumber\\
&\hspace{-10mm}\begin{pmatrix}
0 & 0 & 0 & 0 \\
0 & (\lambda_2^2+\lambda_3^2) \cos(2z) & -(\lambda_2^2+\lambda_3^2)\sin (2z) & 0 \\
0 & -(\lambda_2^2+\lambda_3^2)\sin (2z) & -(\lambda_2^2+\lambda_3^2) \cos(2z) & 0 \\
0 & 0 & 0 & -(\lambda_2^2-\lambda_3^2) \\
\end{pmatrix}\!,
\end{align}
so $f_{\mu\nu}$ is not an (A)dS metric. Hence it is not maximally symmetric and does not have the same isometry group as $g_{\mu\nu}$. 

To summarize, due to screening from the stress--energy tensors, the tensor potentials of $g_{\mu\nu}$ and $f_{\mu\nu}$ decouple from the field equations. The interaction between the two metrics vanishes, i.e., the metrics are effectively non-interacting. Hence it is not surprising that $g_{\mu\nu}$ and $f_{\mu\nu}$ could have completely different isometry groups. A note on the stress--energy tensor of $f_{\mu\nu}$ might also be in place. Since $T_f{}^\mu{}_\nu$ is defined to be proportional to the Einstein tensor for $G_f{}^\mu{}_\nu$, we are not guaranteed that such a stress-energy tensor would satisfy, e.g., the energy conditions. In fact, as shown in \cite{Baccetti:2012re}, in HR bimetric theory in vacuum, there is a strong anti-correlation between the null energy conditions in the $g$- and $f$-sector. However, it might be argued that from an observational viewpoint, the $f$-sector is insignificant as long as the $g_{\mu\nu}$ sector is well-behaved since only $g_{\mu\nu}$ and $T_g{}^\mu{}_\nu$ are measurable.

\subsection*{Method 3: Non-GR solutions}

Many works in the literature studied particular solutions of HR bimetric theory, both 
in vacuum and with matter sources (for a review, see e.g. 
\cite{Schmidt-May:2015vnx} and references therein), and the usual approach is 
to choose some ansatz for the two metrics, motivated by an assumption about 
their isometries. The approach that we follow here is essentially the same, but now 
we allow for different KVFs in the two sectors. In particular, we assume a 
specific form of the map $A{}^\mu{}_\nu$ between the KVFs, and we solve 
\eqref{eq:LieSimple} for $f_{\mu\nu}$ (or, equivalently, $S{}^\mu{}_\nu$). Afterwards, we check that 
the obtained ansatz satisfies the bimetric field equations \eqref{eq:EoM}. If so, then the determined ansatz describes a solution of the theory, otherwise such a configuration is not allowed in the theory.

We have discussed the possibility of having the same isometry group, but different KVFs. In this case, the two metrics do not look explicitly symmetric in the same coordinate chart. Indeed, we can only choose one set of generators of the Lie algebra of the isometry group, and if we choose the coordinate basis to be aligned with some of the generators of the Lie algebra, then the components of another set of generators will not look as simple in the same coordinate basis.

The first potential ansatz we consider is obtained by choosing, in \eqref{eq:map},
\begin{equation}
\varphi{}^\mu{}_\nu=\delta {}^\mu{}_\nu, \qquad A{}^\mu{}_\nu=\gamma {}^\mu{}_\nu=f^{\mu\nu}g_{\mu\nu}=\left(S^{-2}\right){}^\mu{}_\nu.
\end{equation}
In \autoref{subsec:particularCases}, we saw that in this case, the majority of the cases (see \cite{Hall:1983a,Hall:1983b,hall2004symmetries}) provides two proportional metrics, because the two metrics have to share their compatible covariant derivative (see \eqref{eq:sameCovectors}). In HR bimetric theory, if the metrics are proportional, then the field equations \eqref{eq:EoM} reduce to two decoupled copies of the Einstein equations, i.e., $g_{\mu\nu}$ is a solution of the Einstein equations and $f_{\mu\nu}=c^2g_{\mu\nu}$, with $c\in \mathbb{R}$; hence $S{}^\mu{}_\nu=c \, \delta {}^\mu{}_\nu$. We conclude that, assuming $A{}^\mu{}_\nu=\left(S^{-2}\right){}^\mu{}_\nu$, the ansatz having proportional metrics and proportional KVFs is consistent.

Note that $\left(S^{-2}\right){}^\mu{}_\nu$ is not a congruence between $g_{\mu\nu}$ and $f_{\mu\nu}$ in HR bimetric theory. Therefore, in this case, we are not guaranteed that \eqref{eq:sameCharacter1} holds; however, $A{}^\mu{}_\nu$ being a congruence between $g_{\mu\nu}$ and $f_{\mu\nu}$ is a sufficient, but not necessary condition for the KVFs of the two sectors to have the same character. Therefore, there are solutions for which the KVFs are not related by a congruence between the metrics, but still have \arel KVFs with the same character in the two sectors.

The second potential ansatz we consider is,
\begin{align}
\varphi{}^\mu{}_\nu=\delta {}^\mu{}_\nu, \qquad A{}^\mu{}_\nu=\gamma {}^\mu{}_\nu=\left(S^{-1}\right){}^\mu{}_\nu.
\end{align}
for which we are assured that paired KVFs have the same character with respect to the corresponding metric. Indeed, in HR bimetric theory, by definition we are guaranteed that the square root $S{}^\mu{}_\nu=\sqrt{g^{-1}f}\,{}^\mu{}_\nu$ exists \cite{Hassan:2017ugh}, and thus (in matrix notation),
\begin{equation}
f=S^\tr g \, S.
\end{equation}
On the other hand, in this case \eqref{eq:LieSimple} is not as simple as the previous case, since
\begin{equation}
\Upsilon =fS^{-1}g^{-1}= f\sqrt{f^{-1}g}g^{-1}=S^\tr,
\end{equation}
which implies (coming back to the index notation),
\begin{align}
\label{eq:LieSTranspose}
\mathscr{L}_{\xi}f_{\mu\nu} & =\left[ 
2\nabla_{(\mu}{{S^\tr}_{\nu)}}^\beta-{C^\alpha}_{\mu\nu}{{S^\tr}_\alpha}^\beta\right] \eta_\beta \nonumber \\
						&\quad \, +2{{S^\tr}_{(\nu|}}^\beta\left(\nabla_{|\mu)}\eta_\beta \right)=0.
\end{align}

We report here the solution of \eqref{eq:LieSTranspose} in the case of a static and spherically symmetric $g_{\mu\nu}$ in the spherical polar chart $(t,r,\theta,\phi)$,
\begin{align}
\label{eq:SSSg}
g&=
	 \begin{pmatrix}
	   	-\ee ^q(r)F(r) & 0 & 0 & 0 \\
	   	0 & F(r)^{-1} & 0 & 0 \\
	   	0 & 0 & r^2 & 0 \\
	   	0 & 0 & 0 & r^2 \sin^2 (\theta ) \\
	  \end{pmatrix}
\end{align}
with the usual set of KVFs, assuming a handy expression for $S{}^\mu{}_\nu$. We first determine the most general expression for $S{}^\mu{}_\nu$, obtained by the constraint that $h_{\mu\nu}=g_{\mu\rho}S{}^\rho{}_\nu$ (\emph{the symmetrizing quadratic form}) is symmetric; the symmetrisation of $h_{\mu\nu}$ automatically implies the same for $f_{\mu\nu}$. The resulting $S{}^\mu{}_\nu$ is,
\begin{align}
\label{eq:SSSS}
S{}^\mu{}_\nu&=
	 \begin{pmatrix}[1.45]
	   	S_{00} & S_{01} & S_{02} & S_{03} \\
	   	-\ee^qF^2S_{01} & S_{11} & S_{12} & S_{13} \\[0.35mm]
	   	-\dfrac{\ee^qF}{r^2}S_{02} & \dfrac{S_{12}}{Fr^2} & S_{22} & S_{23} \\[2.25mm]
	   	-\dfrac{\ee^qF}{r^2\sin^2 (\theta)}S_{03} & \dfrac{S_{13}}{Fr^2\sin^2 (\theta)} & \dfrac{S_{23}}{\sin^2 (\theta)} & S_{33} \\
	  \end{pmatrix}\!,
\end{align}
where all the functions $S_{ij}$ depend on all the variables $(t,r,\theta,\phi)$. Since solving \eqref{eq:LieSTranspose} with the general ansatz \eqref{eq:SSSS} turns out to be quite involved, we assumed a simpler form for $S{}^\mu{}_\nu$, setting to 0 some of its components. Here we report only one case which turns out, in the end, to have proportional KVFs, with

Let's assume the following $S{}^\mu{}_\nu$,
\begin{align}
\label{eq:SSSSsimple}
S{}^\mu{}_\nu&=
	 \begin{pmatrix}
	   	S_{00} & S_{01} & 0 & 0 \\
	   	-\ee^qF^2S_{01} & S_{11} & 0 & 0 \\
	   	0 & 0 & S_{22} & 0 \\
	   	0 & 0 & 0 & S_{33} \\
	  \end{pmatrix}\!,
\end{align}
\noindent where the functions $S_{ij}$ depend on all the variables. Solving \eqref{eq:LieSTranspose} yields,
\begin{align}
\label{eq:SSSSsol}
S{}^\mu{}_\nu&=
	 \begin{pmatrix}
	   	S_{00}(t) & 0 & 0 & 0 \\
	   0 & S_{11}(r) & 0 & 0 \\
	   	0 & 0 & \lambda & 0 \\
	   	0 & 0 & 0 & \lambda \\
	  \end{pmatrix}\!,
\end{align}
with $\lambda= \mathrm{const}$. Given that $\xi^\mu=\left(S^{-1}\right){}^\mu{}_\nu\eta^\nu$, we conclude that the $S^{-1}$-related KVFs are proportional in this framework.

The field equations \eqref{eq:EoM} and the Bianchi constraint \eqref{eq:Bianchi} require $S_{00}(t)=\mathrm{const}$. This ansatz only contains three field variables, namely $q(r)$, $F(r)$ and $S_{11}(r)$, but the field equations and the Bianchi constraints constitute a system of seven coupled ordinary differential equations. Therefore, the system is over-determined and after solving analytically for the three fields, there are algebraic equations and consistency relations to be satisfied. The solution of this differential--algebraic system (analytical or numerical) is left for future work.
\vspace{0.5cm}

\section{Alternative proof for \eqref{eq:Komar}}
\label{app:alternativeproof}

The Komar identity, presented in \cite{PhysRev.113.934} and proved in \cite{doi:10.1063/1.1705011}, states that, for a generic vector field $V^\mu$,
\begin{equation}
\label{eq:Komar2}
\nabla _\mu \nabla _\nu \nabla ^{[\nu}V^{\mu]} \equiv 0.
\end{equation}
This a purely geometrical identity, therefore it holds in any modified theory of gravity relying on (pseudo)-Riemannian geometry.

The proof in \cite{doi:10.1063/1.1705011} is based on the variation of the Einstein--Hilbert Lagrangian density. We now present an alternative and simpler proof.

If $\nabla ^{[\nu}V^{\mu]}=0$, then the Komar identity is trivially true, hence we suppose it is not zero. By making use of the commutation rule of covariant derivatives involving the Riemann curvature tensor \cite[p.~39]{Wald:1984rg}, we can rewrite \eqref{eq:Komar2} as
\begin{align}
\label{eq:commutations}
								&\nabla _\mu \nabla _\nu \nabla ^{[\nu}V^{\mu]} \nonumber \\
								&= \nabla _\mu \nabla _\nu \nabla ^\nu V^\mu-\nabla _\mu \nabla _\nu \nabla ^\mu V^\nu \nonumber \\
								&=\nabla _\nu \nabla _\mu \nabla ^\nu V^\mu-R_{\mu\nu\rho}{}^\nu\nabla ^\rho V^\mu-R_{\mu \nu \sigma}{}^\mu\nabla ^\nu V^\sigma \nonumber \\
								&\quad-\nabla _\nu \nabla _\mu \nabla ^\mu V^\nu+R_{\mu\nu\rho}{}^\mu\nabla ^\rho V^\nu+R_{\mu \nu \sigma}{}^\nu\nabla ^\mu V^\sigma \nonumber \\
								&=\nabla _\nu \nabla _\mu \nabla ^\nu V^\mu-R_{\mu\rho}\nabla ^\rho V^\mu+R_{\nu \sigma}\nabla ^\nu V^\sigma \nonumber \\
								&\quad-\nabla _\nu \nabla _\mu \nabla ^\mu V^\nu-R_{\nu\sigma}\nabla ^\sigma V^\nu+R_{\mu \rho}\nabla ^\mu V^\rho \nonumber \\
								&=\nabla _\nu \nabla _\mu \nabla ^{[\nu} V^{\mu]}+R_{\mu\rho}\nabla ^{[\mu} V^{\rho]}+R_{\nu\sigma}\nabla ^{[\nu} V^{\sigma]} \nonumber \\
								&=\nabla _\nu \nabla _\mu \nabla ^{[\nu}V^{\mu]},
\end{align}
i.e.,
\begin{equation}
\label{eq:anti0}
\nabla _\mu \nabla _\nu \nabla ^{[\nu}V^{\mu]}-\nabla _\nu \nabla _\mu \nabla ^{[\nu}V^{\mu]}=2\nabla _{[\mu} \nabla _{\nu]} \nabla ^{[\nu}V^{\mu]}\equiv 0.
\end{equation}
From \eqref{eq:anti0} it follows
\begin{equation}
\label{eq:simm}
\nabla _\mu \nabla _\nu \nabla ^{[\nu}V^{\mu]}=\nabla _{(\mu} \nabla _{\nu)} \nabla ^{[\nu}V^{\mu]}\equiv 0,
\end{equation}
since we are contracting a symmetric expression in $\mu, \nu$, with an antisymmetric expression in $\mu,\nu$. This completes the proof.

For a KVF, $\nabla ^{[\nu}\xi^{\mu]}=\nabla ^{\nu}\xi^{\mu}$ due to the Killing equation \eqref{eq:Kequation}, therefore
\begin{equation}
\label{eq:equivalence}
\nabla _\mu \nabla _\nu \nabla ^{[\nu}\xi^{\mu]} = \nabla _\mu \left(\nabla ^2 \xi^\mu \right)\equiv 0.
\end{equation}

\newpage

\bibliographystyle{apsrev4-1}

\begin{thebibliography}{52}%
\makeatletter
\providecommand \@ifxundefined [1]{%
 \@ifx{#1\undefined}
}%
\providecommand \@ifnum [1]{%
 \ifnum #1\expandafter \@firstoftwo
 \else \expandafter \@secondoftwo
 \fi
}%
\providecommand \@ifx [1]{%
 \ifx #1\expandafter \@firstoftwo
 \else \expandafter \@secondoftwo
 \fi
}%
\providecommand \natexlab [1]{#1}%
\providecommand \enquote  [1]{``#1''}%
\providecommand \bibnamefont  [1]{#1}%
\providecommand \bibfnamefont [1]{#1}%
\providecommand \citenamefont [1]{#1}%
\providecommand \href@noop [0]{\@secondoftwo}%
\providecommand \href [0]{\begingroup \@sanitize@url \@href}%
\providecommand \@href[1]{\@@startlink{#1}\@@href}%
\providecommand \@@href[1]{\endgroup#1\@@endlink}%
\providecommand \@sanitize@url [0]{\catcode `\\12\catcode `\$12\catcode
  `\&12\catcode `\#12\catcode `\^12\catcode `\_12\catcode `\%12\relax}%
\providecommand \@@startlink[1]{}%
\providecommand \@@endlink[0]{}%
\providecommand \url  [0]{\begingroup\@sanitize@url \@url }%
\providecommand \@url [1]{\endgroup\@href {#1}{\urlprefix }}%
\providecommand \urlprefix  [0]{URL }%
\providecommand \Eprint [0]{\href }%
\providecommand \doibase [0]{http://dx.doi.org/}%
\providecommand \selectlanguage [0]{\@gobble}%
\providecommand \bibinfo  [0]{\@secondoftwo}%
\providecommand \bibfield  [0]{\@secondoftwo}%
\providecommand \translation [1]{[#1]}%
\providecommand \BibitemOpen [0]{}%
\providecommand \bibitemStop [0]{}%
\providecommand \bibitemNoStop [0]{.\EOS\space}%
\providecommand \EOS [0]{\spacefactor3000\relax}%
\providecommand \BibitemShut  [1]{\csname bibitem#1\endcsname}%
\let\auto@bib@innerbib\@empty
\bibitem [{\citenamefont {Schwarzschild}(1916)}]{Schwarzschild:1916uq}%
  \BibitemOpen
  \bibfield  {author} {\bibinfo {author} {\bibfnamefont {K.}~\bibnamefont
  {Schwarzschild}},\ }\href@noop {} {\bibfield  {journal} {\bibinfo  {journal}
  {Sitzungsber. Preuss. Akad. Wiss. Berlin (Math. Phys.)}\ }\textbf {\bibinfo
  {volume} {1916}},\ \bibinfo {pages} {189} (\bibinfo {year} {1916})},\ \Eprint
  {http://arxiv.org/abs/physics/9905030} {arXiv:physics/9905030 [physics]}
  \BibitemShut {NoStop}%
\bibitem [{\citenamefont {Wald}(2010)}]{Wald:1984rg}%
  \BibitemOpen
  \bibfield  {author} {\bibinfo {author} {\bibfnamefont {R.}~\bibnamefont
  {Wald}},\ }\href {https://books.google.se/books?id=9S-hzg6-moYC} {\emph
  {\bibinfo {title} {General Relativity}}}\ (\bibinfo  {publisher} {University
  of Chicago Press},\ \bibinfo {year} {2010})\BibitemShut {NoStop}%
\bibitem [{\citenamefont {Hall}(2004)}]{hall2004symmetries}%
  \BibitemOpen
  \bibfield  {author} {\bibinfo {author} {\bibfnamefont {G.}~\bibnamefont
  {Hall}},\ }\href {https://books.google.se/books?id=VLHsCgAAQBAJ} {\emph
  {\bibinfo {title} {Symmetries and Curvature Structure in General
  Relativity}}},\ World Scientific Lecture Notes in Physics\ (\bibinfo {year}
  {2004})\BibitemShut {NoStop}%
\bibitem [{\citenamefont {Volkov}(2012)}]{PhysRevD.86.061502}%
  \BibitemOpen
  \bibfield  {author} {\bibinfo {author} {\bibfnamefont {M.~S.}\ \bibnamefont
  {Volkov}},\ }\href {\doibase 10.1103/PhysRevD.86.061502} {\bibfield
  {journal} {\bibinfo  {journal} {Phys. Rev. D}\ }\textbf {\bibinfo {volume}
  {86}},\ \bibinfo {pages} {061502} (\bibinfo {year} {2012})}\BibitemShut
  {NoStop}%
\bibitem [{\citenamefont {Nersisyan}\ \emph {et~al.}(2015)\citenamefont
  {Nersisyan}, \citenamefont {Akrami},\ and\ \citenamefont
  {Amendola}}]{Nersisyan:2015oha}%
  \BibitemOpen
  \bibfield  {author} {\bibinfo {author} {\bibfnamefont {H.}~\bibnamefont
  {Nersisyan}}, \bibinfo {author} {\bibfnamefont {Y.}~\bibnamefont {Akrami}}, \
  and\ \bibinfo {author} {\bibfnamefont {L.}~\bibnamefont {Amendola}},\ }\href
  {\doibase 10.1103/PhysRevD.92.104034} {\bibfield  {journal} {\bibinfo
  {journal} {Phys. Rev.}\ }\textbf {\bibinfo {volume} {D92}},\ \bibinfo {pages}
  {104034} (\bibinfo {year} {2015})},\ \Eprint
  {http://arxiv.org/abs/1502.03988} {arXiv:1502.03988 [gr-qc]} \BibitemShut
  {NoStop}%
\bibitem [{\citenamefont {Boulware}\ and\ \citenamefont
  {Deser}(1972)}]{Boulware:1973my}%
  \BibitemOpen
  \bibfield  {author} {\bibinfo {author} {\bibfnamefont {D.~G.}\ \bibnamefont
  {Boulware}}\ and\ \bibinfo {author} {\bibfnamefont {S.}~\bibnamefont
  {Deser}},\ }\href {\doibase 10.1103/PhysRevD.6.3368} {\bibfield  {journal}
  {\bibinfo  {journal} {Phys. Rev.}\ }\textbf {\bibinfo {volume} {D6}},\
  \bibinfo {pages} {3368} (\bibinfo {year} {1972})}\BibitemShut {NoStop}%
\bibitem [{\citenamefont {de~Rham}\ and\ \citenamefont
  {Gabadadze}(2010)}]{deRham:2010ik}%
  \BibitemOpen
  \bibfield  {author} {\bibinfo {author} {\bibfnamefont {C.}~\bibnamefont
  {de~Rham}}\ and\ \bibinfo {author} {\bibfnamefont {G.}~\bibnamefont
  {Gabadadze}},\ }\href {\doibase 10.1103/PhysRevD.82.044020} {\bibfield
  {journal} {\bibinfo  {journal} {Phys. Rev.}\ }\textbf {\bibinfo {volume}
  {D82}},\ \bibinfo {pages} {044020} (\bibinfo {year} {2010})},\ \Eprint
  {http://arxiv.org/abs/1007.0443} {arXiv:1007.0443 [hep-th]} \BibitemShut
  {NoStop}%
\bibitem [{\citenamefont {de~Rham}\ \emph {et~al.}(2011)\citenamefont
  {de~Rham}, \citenamefont {Gabadadze},\ and\ \citenamefont
  {Tolley}}]{deRham:2010kj}%
  \BibitemOpen
  \bibfield  {author} {\bibinfo {author} {\bibfnamefont {C.}~\bibnamefont
  {de~Rham}}, \bibinfo {author} {\bibfnamefont {G.}~\bibnamefont {Gabadadze}},
  \ and\ \bibinfo {author} {\bibfnamefont {A.~J.}\ \bibnamefont {Tolley}},\
  }\href {\doibase 10.1103/PhysRevLett.106.231101} {\bibfield  {journal}
  {\bibinfo  {journal} {Phys. Rev. Lett.}\ }\textbf {\bibinfo {volume} {106}},\
  \bibinfo {pages} {231101} (\bibinfo {year} {2011})},\ \Eprint
  {http://arxiv.org/abs/1011.1232} {arXiv:1011.1232 [hep-th]} \BibitemShut
  {NoStop}%
\bibitem [{\citenamefont {Hassan}\ and\ \citenamefont
  {Rosen}(2012{\natexlab{a}})}]{Hassan:2011hr}%
  \BibitemOpen
  \bibfield  {author} {\bibinfo {author} {\bibfnamefont {S.~F.}\ \bibnamefont
  {Hassan}}\ and\ \bibinfo {author} {\bibfnamefont {R.~A.}\ \bibnamefont
  {Rosen}},\ }\href {\doibase 10.1103/PhysRevLett.108.041101} {\bibfield
  {journal} {\bibinfo  {journal} {Phys. Rev. Lett.}\ }\textbf {\bibinfo
  {volume} {108}},\ \bibinfo {pages} {041101} (\bibinfo {year}
  {2012}{\natexlab{a}})},\ \Eprint {http://arxiv.org/abs/1106.3344}
  {arXiv:1106.3344 [hep-th]} \BibitemShut {NoStop}%
\bibitem [{\citenamefont {Hassan}\ and\ \citenamefont
  {Rosen}(2012{\natexlab{b}})}]{Hassan:2011zd}%
  \BibitemOpen
  \bibfield  {author} {\bibinfo {author} {\bibfnamefont {S.~F.}\ \bibnamefont
  {Hassan}}\ and\ \bibinfo {author} {\bibfnamefont {R.~A.}\ \bibnamefont
  {Rosen}},\ }\href {\doibase 10.1007/JHEP02(2012)126} {\bibfield  {journal}
  {\bibinfo  {journal} {JHEP}\ }\textbf {\bibinfo {volume} {02}},\ \bibinfo
  {pages} {126} (\bibinfo {year} {2012}{\natexlab{b}})},\ \Eprint
  {http://arxiv.org/abs/1109.3515} {arXiv:1109.3515 [hep-th]} \BibitemShut
  {NoStop}%
\bibitem [{\citenamefont {Hassan}\ and\ \citenamefont
  {Rosen}(2012{\natexlab{c}})}]{Hassan:2011ea}%
  \BibitemOpen
  \bibfield  {author} {\bibinfo {author} {\bibfnamefont {S.~F.}\ \bibnamefont
  {Hassan}}\ and\ \bibinfo {author} {\bibfnamefont {R.~A.}\ \bibnamefont
  {Rosen}},\ }\href {\doibase 10.1007/JHEP04(2012)123} {\bibfield  {journal}
  {\bibinfo  {journal} {JHEP}\ }\textbf {\bibinfo {volume} {04}},\ \bibinfo
  {pages} {123} (\bibinfo {year} {2012}{\natexlab{c}})},\ \Eprint
  {http://arxiv.org/abs/1111.2070} {arXiv:1111.2070 [hep-th]} \BibitemShut
  {NoStop}%
\bibitem [{\citenamefont {Hassan}\ and\ \citenamefont
  {Kocic}(2017)}]{Hassan:2017ugh}%
  \BibitemOpen
  \bibfield  {author} {\bibinfo {author} {\bibfnamefont {S.~F.}\ \bibnamefont
  {Hassan}}\ and\ \bibinfo {author} {\bibfnamefont {M.}~\bibnamefont {Kocic}},\
  }\href@noop {} {\  (\bibinfo {year} {2017})},\ \Eprint
  {http://arxiv.org/abs/1706.07806} {arXiv:1706.07806 [hep-th]} \BibitemShut
  {NoStop}%
\bibitem [{\citenamefont {Comelli}\ \emph {et~al.}(2012)\citenamefont
  {Comelli}, \citenamefont {Crisostomi}, \citenamefont {Nesti},\ and\
  \citenamefont {Pilo}}]{Comelli:2011wq}%
  \BibitemOpen
  \bibfield  {author} {\bibinfo {author} {\bibfnamefont {D.}~\bibnamefont
  {Comelli}}, \bibinfo {author} {\bibfnamefont {M.}~\bibnamefont {Crisostomi}},
  \bibinfo {author} {\bibfnamefont {F.}~\bibnamefont {Nesti}}, \ and\ \bibinfo
  {author} {\bibfnamefont {L.}~\bibnamefont {Pilo}},\ }\href {\doibase
  10.1103/PhysRevD.85.024044} {\bibfield  {journal} {\bibinfo  {journal} {Phys.
  Rev.}\ }\textbf {\bibinfo {volume} {D85}},\ \bibinfo {pages} {024044}
  (\bibinfo {year} {2012})},\ \Eprint {http://arxiv.org/abs/1110.4967}
  {arXiv:1110.4967 [hep-th]} \BibitemShut {NoStop}%
\bibitem [{\citenamefont {Hassan}\ \emph
  {et~al.}(2013{\natexlab{a}})\citenamefont {Hassan}, \citenamefont
  {Schmidt-May},\ and\ \citenamefont {von Strauss}}]{Hassan:2012wr}%
  \BibitemOpen
  \bibfield  {author} {\bibinfo {author} {\bibfnamefont {S.~F.}\ \bibnamefont
  {Hassan}}, \bibinfo {author} {\bibfnamefont {A.}~\bibnamefont {Schmidt-May}},
  \ and\ \bibinfo {author} {\bibfnamefont {M.}~\bibnamefont {von Strauss}},\
  }\href {\doibase 10.1007/JHEP05(2013)086} {\bibfield  {journal} {\bibinfo
  {journal} {JHEP}\ }\textbf {\bibinfo {volume} {05}},\ \bibinfo {pages} {086}
  (\bibinfo {year} {2013}{\natexlab{a}})},\ \Eprint
  {http://arxiv.org/abs/1208.1515} {arXiv:1208.1515 [hep-th]} \BibitemShut
  {NoStop}%
\bibitem [{\citenamefont {Hassan}\ \emph
  {et~al.}(2013{\natexlab{b}})\citenamefont {Hassan}, \citenamefont
  {Schmidt-May},\ and\ \citenamefont {von Strauss}}]{Hassan:2012rq}%
  \BibitemOpen
  \bibfield  {author} {\bibinfo {author} {\bibfnamefont {S.~F.}\ \bibnamefont
  {Hassan}}, \bibinfo {author} {\bibfnamefont {A.}~\bibnamefont {Schmidt-May}},
  \ and\ \bibinfo {author} {\bibfnamefont {M.}~\bibnamefont {von Strauss}},\
  }\href {\doibase 10.1088/0264-9381/30/18/184010} {\bibfield  {journal}
  {\bibinfo  {journal} {Class. Quant. Grav.}\ }\textbf {\bibinfo {volume}
  {30}},\ \bibinfo {pages} {184010} (\bibinfo {year} {2013}{\natexlab{b}})},\
  \Eprint {http://arxiv.org/abs/1212.4525} {arXiv:1212.4525 [hep-th]}
  \BibitemShut {NoStop}%
\bibitem [{\citenamefont {de~Rham}\ \emph {et~al.}(2015)\citenamefont
  {de~Rham}, \citenamefont {Heisenberg},\ and\ \citenamefont
  {Ribeiro}}]{deRham:2014naa}%
  \BibitemOpen
  \bibfield  {author} {\bibinfo {author} {\bibfnamefont {C.}~\bibnamefont
  {de~Rham}}, \bibinfo {author} {\bibfnamefont {L.}~\bibnamefont {Heisenberg}},
  \ and\ \bibinfo {author} {\bibfnamefont {R.~H.}\ \bibnamefont {Ribeiro}},\
  }\href {\doibase 10.1088/0264-9381/32/3/035022} {\bibfield  {journal}
  {\bibinfo  {journal} {Class. Quant. Grav.}\ }\textbf {\bibinfo {volume}
  {32}},\ \bibinfo {pages} {035022} (\bibinfo {year} {2015})},\ \Eprint
  {http://arxiv.org/abs/1408.1678} {arXiv:1408.1678 [hep-th]} \BibitemShut
  {NoStop}%
\bibitem [{\citenamefont {Hassan}\ \emph {et~al.}(2014)\citenamefont {Hassan},
  \citenamefont {Schmidt-May},\ and\ \citenamefont {von
  Strauss}}]{Hassan:2014vja}%
  \BibitemOpen
  \bibfield  {author} {\bibinfo {author} {\bibfnamefont {S.~F.}\ \bibnamefont
  {Hassan}}, \bibinfo {author} {\bibfnamefont {A.}~\bibnamefont {Schmidt-May}},
  \ and\ \bibinfo {author} {\bibfnamefont {M.}~\bibnamefont {von Strauss}},\
  }\href {\doibase 10.1142/S0218271814430020} {\bibfield  {journal} {\bibinfo
  {journal} {Int. J. Mod. Phys.}\ }\textbf {\bibinfo {volume} {D23}},\ \bibinfo
  {pages} {1443002} (\bibinfo {year} {2014})},\ \Eprint
  {http://arxiv.org/abs/1407.2772} {arXiv:1407.2772 [hep-th]} \BibitemShut
  {NoStop}%
\bibitem [{\citenamefont {Damour}\ and\ \citenamefont
  {Kogan}(2002)}]{Damour:2002ws}%
  \BibitemOpen
  \bibfield  {author} {\bibinfo {author} {\bibfnamefont {T.}~\bibnamefont
  {Damour}}\ and\ \bibinfo {author} {\bibfnamefont {I.~I.}\ \bibnamefont
  {Kogan}},\ }\href {\doibase 10.1103/PhysRevD.66.104024} {\bibfield  {journal}
  {\bibinfo  {journal} {Phys. Rev.}\ }\textbf {\bibinfo {volume} {D66}},\
  \bibinfo {pages} {104024} (\bibinfo {year} {2002})},\ \Eprint
  {http://arxiv.org/abs/hep-th/0206042} {arXiv:hep-th/0206042 [hep-th]}
  \BibitemShut {NoStop}%
\bibitem [{\citenamefont {Kocic}\ \emph
  {et~al.}(2017{\natexlab{a}})\citenamefont {Kocic}, \citenamefont {Högås},
  \citenamefont {Torsello},\ and\ \citenamefont {Mortsell}}]{Kocic:2017wwf}%
  \BibitemOpen
  \bibfield  {author} {\bibinfo {author} {\bibfnamefont {M.}~\bibnamefont
  {Kocic}}, \bibinfo {author} {\bibfnamefont {M.}~\bibnamefont {Högås}},
  \bibinfo {author} {\bibfnamefont {F.}~\bibnamefont {Torsello}}, \ and\
  \bibinfo {author} {\bibfnamefont {E.}~\bibnamefont {Mortsell}},\ }\href@noop
  {} {\  (\bibinfo {year} {2017}{\natexlab{a}})},\ \Eprint
  {http://arxiv.org/abs/1706.00787} {arXiv:1706.00787 [hep-th]} \BibitemShut
  {NoStop}%
\bibitem [{\citenamefont {Schmidt-May}\ and\ \citenamefont {von
  Strauss}(2016)}]{Schmidt-May:2015vnx}%
  \BibitemOpen
  \bibfield  {author} {\bibinfo {author} {\bibfnamefont {A.}~\bibnamefont
  {Schmidt-May}}\ and\ \bibinfo {author} {\bibfnamefont {M.}~\bibnamefont {von
  Strauss}},\ }\href {\doibase 10.1088/1751-8113/49/18/183001} {\bibfield
  {journal} {\bibinfo  {journal} {J. Phys.}\ }\textbf {\bibinfo {volume}
  {A49}},\ \bibinfo {pages} {183001} (\bibinfo {year} {2016})},\ \Eprint
  {http://arxiv.org/abs/1512.00021} {arXiv:1512.00021 [hep-th]} \BibitemShut
  {NoStop}%
\bibitem [{\citenamefont {Boulanger}\ \emph {et~al.}(2001)\citenamefont
  {Boulanger}, \citenamefont {Damour}, \citenamefont {Gualtieri},\ and\
  \citenamefont {Henneaux}}]{Boulanger:2000rq}%
  \BibitemOpen
  \bibfield  {author} {\bibinfo {author} {\bibfnamefont {N.}~\bibnamefont
  {Boulanger}}, \bibinfo {author} {\bibfnamefont {T.}~\bibnamefont {Damour}},
  \bibinfo {author} {\bibfnamefont {L.}~\bibnamefont {Gualtieri}}, \ and\
  \bibinfo {author} {\bibfnamefont {M.}~\bibnamefont {Henneaux}},\ }\href
  {\doibase 10.1016/S0550-3213(00)00718-5} {\bibfield  {journal} {\bibinfo
  {journal} {Nucl. Phys.}\ }\textbf {\bibinfo {volume} {B597}},\ \bibinfo
  {pages} {127} (\bibinfo {year} {2001})},\ \Eprint
  {http://arxiv.org/abs/hep-th/0007220} {arXiv:hep-th/0007220 [hep-th]}
  \BibitemShut {NoStop}%
\bibitem [{\citenamefont {Lang}(1995)}]{lang1995differential}%
  \BibitemOpen
  \bibfield  {author} {\bibinfo {author} {\bibfnamefont {S.}~\bibnamefont
  {Lang}},\ }\href {https://books.google.se/books?id=VfxGB5nYv1MC} {\emph
  {\bibinfo {title} {Differential and Riemannian Manifolds}}},\ Graduate Texts
  in Mathematics\ (\bibinfo  {publisher} {Springer},\ \bibinfo {year}
  {1995})\BibitemShut {NoStop}%
\bibitem [{\citenamefont {Misner}\ \emph {et~al.}(1973)\citenamefont {Misner},
  \citenamefont {Thorne},\ and\ \citenamefont
  {Wheeler}}]{misner1973gravitation}%
  \BibitemOpen
  \bibfield  {author} {\bibinfo {author} {\bibfnamefont {C.}~\bibnamefont
  {Misner}}, \bibinfo {author} {\bibfnamefont {K.}~\bibnamefont {Thorne}}, \
  and\ \bibinfo {author} {\bibfnamefont {J.}~\bibnamefont {Wheeler}},\ }\href
  {https://books.google.se/books?id=w4Gigq3tY1kC} {\emph {\bibinfo {title}
  {Gravitation}}},\ \bibinfo {series} {Gravitation}\ No.\ \bibinfo {number}
  {pt. 3}\ (\bibinfo  {publisher} {W. H. Freeman},\ \bibinfo {year}
  {1973})\BibitemShut {NoStop}%
\bibitem [{\citenamefont {Fuller}\ and\ \citenamefont
  {Wheeler}(1962)}]{PhysRev.128.919}%
  \BibitemOpen
  \bibfield  {author} {\bibinfo {author} {\bibfnamefont {R.~W.}\ \bibnamefont
  {Fuller}}\ and\ \bibinfo {author} {\bibfnamefont {J.~A.}\ \bibnamefont
  {Wheeler}},\ }\href {\doibase 10.1103/PhysRev.128.919} {\bibfield  {journal}
  {\bibinfo  {journal} {Phys. Rev.}\ }\textbf {\bibinfo {volume} {128}},\
  \bibinfo {pages} {919} (\bibinfo {year} {1962})}\BibitemShut {NoStop}%
\bibitem [{\citenamefont {Synge}(1960)}]{synge1960relativity}%
  \BibitemOpen
  \bibfield  {author} {\bibinfo {author} {\bibfnamefont {J.}~\bibnamefont
  {Synge}},\ }\href {https://books.google.se/books?id=CqoNAQAAIAAJ} {\emph
  {\bibinfo {title} {Relativity: the general theory}}},\ Series in physics\
  (\bibinfo  {publisher} {North-Holland Pub. Co.},\ \bibinfo {year}
  {1960})\BibitemShut {NoStop}%
\bibitem [{\citenamefont {Ibohal}(2005)}]{Ibohal:2004kk}%
  \BibitemOpen
  \bibfield  {author} {\bibinfo {author} {\bibfnamefont {N.}~\bibnamefont
  {Ibohal}},\ }\href {\doibase 10.1007/s10714-005-0002-6} {\bibfield  {journal}
  {\bibinfo  {journal} {Gen. Rel. Grav.}\ }\textbf {\bibinfo {volume} {37}},\
  \bibinfo {pages} {19} (\bibinfo {year} {2005})},\ \Eprint
  {http://arxiv.org/abs/gr-qc/0403098} {arXiv:gr-qc/0403098 [gr-qc]}
  \BibitemShut {NoStop}%
\bibitem [{\citenamefont {Kocic}\ \emph
  {et~al.}(2017{\natexlab{b}})\citenamefont {Kocic}, \citenamefont {Högås},
  \citenamefont {Torsello},\ and\ \citenamefont {Mortsell}}]{Kocic:2017hve}%
  \BibitemOpen
  \bibfield  {author} {\bibinfo {author} {\bibfnamefont {M.}~\bibnamefont
  {Kocic}}, \bibinfo {author} {\bibfnamefont {M.}~\bibnamefont {Högås}},
  \bibinfo {author} {\bibfnamefont {F.}~\bibnamefont {Torsello}}, \ and\
  \bibinfo {author} {\bibfnamefont {E.}~\bibnamefont {Mortsell}},\ }\href@noop
  {} {\  (\bibinfo {year} {2017}{\natexlab{b}})},\ \Eprint
  {http://arxiv.org/abs/1708.07833} {arXiv:1708.07833 [hep-th]} \BibitemShut
  {NoStop}%
\bibitem [{\citenamefont {Deffayet}\ and\ \citenamefont
  {Jacobson}(2012)}]{Deffayet:2011rh}%
  \BibitemOpen
  \bibfield  {author} {\bibinfo {author} {\bibfnamefont {C.}~\bibnamefont
  {Deffayet}}\ and\ \bibinfo {author} {\bibfnamefont {T.}~\bibnamefont
  {Jacobson}},\ }\href {\doibase 10.1088/0264-9381/29/6/065009} {\bibfield
  {journal} {\bibinfo  {journal} {Class. Quant. Grav.}\ }\textbf {\bibinfo
  {volume} {29}},\ \bibinfo {pages} {065009} (\bibinfo {year} {2012})},\
  \Eprint {http://arxiv.org/abs/1107.4978} {arXiv:1107.4978 [gr-qc]}
  \BibitemShut {NoStop}%
\bibitem [{\citenamefont {Torsello}\ \emph
  {et~al.}(2017{\natexlab{a}})\citenamefont {Torsello}, \citenamefont {Kocic},\
  and\ \citenamefont {M\"ortsell}}]{PhysRevD.96.064003}%
  \BibitemOpen
  \bibfield  {author} {\bibinfo {author} {\bibfnamefont {F.}~\bibnamefont
  {Torsello}}, \bibinfo {author} {\bibfnamefont {M.}~\bibnamefont {Kocic}}, \
  and\ \bibinfo {author} {\bibfnamefont {E.}~\bibnamefont {M\"ortsell}},\
  }\href {\doibase 10.1103/PhysRevD.96.064003} {\bibfield  {journal} {\bibinfo
  {journal} {Phys. Rev. D}\ }\textbf {\bibinfo {volume} {96}},\ \bibinfo
  {pages} {064003} (\bibinfo {year} {2017}{\natexlab{a}})},\ \Eprint
  {http://arxiv.org/abs/1703.07787} {arXiv:1703.07787 [gr-qc]} \BibitemShut
  {NoStop}%
\bibitem [{\citenamefont {Moschella}(2006)}]{Moschella2006}%
  \BibitemOpen
  \bibfield  {author} {\bibinfo {author} {\bibfnamefont {U.}~\bibnamefont
  {Moschella}},\ }\enquote {\bibinfo {title} {The de {Sitter} and anti-de
  {Sitter} {Sightseeing} {Tour}},}\ in\ \href {\doibase
  10.1007/3-7643-7436-5_4} {\emph {\bibinfo {booktitle} {Einstein, 1905--2005:
  Poincar{\'e} Seminar 2005}}},\ \bibinfo {editor} {edited by\ \bibinfo
  {editor} {\bibfnamefont {T.}~\bibnamefont {Damour}}, \bibinfo {editor}
  {\bibfnamefont {O.}~\bibnamefont {Darrigol}}, \bibinfo {editor}
  {\bibfnamefont {B.}~\bibnamefont {Duplantier}}, \ and\ \bibinfo {editor}
  {\bibfnamefont {V.}~\bibnamefont {Rivasseau}}}\ (\bibinfo  {publisher}
  {Birkh{\"a}user Basel},\ \bibinfo {address} {Basel},\ \bibinfo {year}
  {2006})\ pp.\ \bibinfo {pages} {120--133}\BibitemShut {NoStop}%
\bibitem [{\citenamefont {{Wolfram Research, Inc.}}(2016)}]{mathematica}%
  \BibitemOpen
  \bibfield  {author} {\bibinfo {author} {\bibnamefont {{Wolfram Research,
  Inc.}}},\ }\href@noop {} {\emph {\bibinfo {title} {{Mathematica, Version
  11.0}}}},\ \bibinfo {address} {Champaign, Illinois} (\bibinfo {year}
  {2016})\BibitemShut {NoStop}%
\bibitem [{\citenamefont {Hawking}\ and\ \citenamefont
  {Ellis}(2011)}]{Hawking:1973uf}%
  \BibitemOpen
  \bibfield  {author} {\bibinfo {author} {\bibfnamefont {S.~W.}\ \bibnamefont
  {Hawking}}\ and\ \bibinfo {author} {\bibfnamefont {G.~F.~R.}\ \bibnamefont
  {Ellis}},\ }\href {\doibase 10.1017/CBO9780511524646} {\emph {\bibinfo
  {title} {{The Large Scale Structure of Space-Time}}}},\ Cambridge Monographs
  on Mathematical Physics\ (\bibinfo  {publisher} {Cambridge University
  Press},\ \bibinfo {year} {2011})\BibitemShut {NoStop}%
\bibitem [{\citenamefont {Gibbons}\ and\ \citenamefont
  {Hawking}(1977)}]{PhysRevD.15.2738}%
  \BibitemOpen
  \bibfield  {author} {\bibinfo {author} {\bibfnamefont {G.~W.}\ \bibnamefont
  {Gibbons}}\ and\ \bibinfo {author} {\bibfnamefont {S.~W.}\ \bibnamefont
  {Hawking}},\ }\href {\doibase 10.1103/PhysRevD.15.2738} {\bibfield  {journal}
  {\bibinfo  {journal} {Phys. Rev. D}\ }\textbf {\bibinfo {volume} {15}},\
  \bibinfo {pages} {2738} (\bibinfo {year} {1977})}\BibitemShut {NoStop}%
\bibitem [{\citenamefont {Charmousis}(2011)}]{Charmousis2011}%
  \BibitemOpen
  \bibfield  {author} {\bibinfo {author} {\bibfnamefont {C.}~\bibnamefont
  {Charmousis}},\ }\enquote {\bibinfo {title} {Introduction to anti de sitter
  black holes},}\ in\ \href {\doibase 10.1007/978-3-642-04864-7_1} {\emph
  {\bibinfo {booktitle} {From Gravity to Thermal Gauge Theories: The AdS/CFT
  Correspondence}}},\ \bibinfo {editor} {edited by\
  \bibinfo {editor} {\bibfnamefont {E.}~\bibnamefont {Papantonopoulos}}}\
  (\bibinfo  {publisher} {Springer Berlin Heidelberg},\ \bibinfo {address}
  {Berlin, Heidelberg},\ \bibinfo {year} {2011})\ pp.\ \bibinfo {pages}
  {3--26}\BibitemShut {NoStop}%
\bibitem [{\citenamefont {Hawking}(1972)}]{Hawking1972}%
  \BibitemOpen
  \bibfield  {author} {\bibinfo {author} {\bibfnamefont {S.~W.}\ \bibnamefont
  {Hawking}},\ }\href {\doibase 10.1007/BF01877517} {\bibfield  {journal}
  {\bibinfo  {journal} {Communications in Mathematical Physics}\ }\textbf
  {\bibinfo {volume} {25}},\ \bibinfo {pages} {152} (\bibinfo {year}
  {1972})}\BibitemShut {NoStop}%
\bibitem [{\citenamefont {Henneaux}\ and\ \citenamefont
  {Teitelboim}(1985)}]{Henneaux1985}%
  \BibitemOpen
  \bibfield  {author} {\bibinfo {author} {\bibfnamefont {M.}~\bibnamefont
  {Henneaux}}\ and\ \bibinfo {author} {\bibfnamefont {C.}~\bibnamefont
  {Teitelboim}},\ }\href {\doibase 10.1007/BF01205790} {\bibfield  {journal}
  {\bibinfo  {journal} {Communications in Mathematical Physics}\ }\textbf
  {\bibinfo {volume} {98}},\ \bibinfo {pages} {391} (\bibinfo {year}
  {1985})}\BibitemShut {NoStop}%
\bibitem [{\citenamefont {Jamal}(2017)}]{Jamal2017}%
  \BibitemOpen
  \bibfield  {author} {\bibinfo {author} {\bibfnamefont {S.}~\bibnamefont
  {Jamal}},\ }\href {\doibase 10.1007/s10714-017-2253-4} {\bibfield  {journal}
  {\bibinfo  {journal} {General Relativity and Gravitation}\ }\textbf {\bibinfo
  {volume} {49}},\ \bibinfo {pages} {88} (\bibinfo {year} {2017})}\BibitemShut
  {NoStop}%
\bibitem [{\citenamefont {Komar}(1959)}]{PhysRev.113.934}%
  \BibitemOpen
  \bibfield  {author} {\bibinfo {author} {\bibfnamefont {A.}~\bibnamefont
  {Komar}},\ }\href {\doibase 10.1103/PhysRev.113.934} {\bibfield  {journal}
  {\bibinfo  {journal} {Phys. Rev.}\ }\textbf {\bibinfo {volume} {113}},\
  \bibinfo {pages} {934} (\bibinfo {year} {1959})}\BibitemShut {NoStop}%
\bibitem [{\citenamefont {Davis}\ and\ \citenamefont
  {Moss}(1966)}]{doi:10.1063/1.1705011}%
  \BibitemOpen
  \bibfield  {author} {\bibinfo {author} {\bibfnamefont {W.~R.}\ \bibnamefont
  {Davis}}\ and\ \bibinfo {author} {\bibfnamefont {M.~K.}\ \bibnamefont
  {Moss}},\ }\href {\doibase 10.1063/1.1705011} {\bibfield  {journal} {\bibinfo
   {journal} {Journal of Mathematical Physics}\ }\textbf {\bibinfo {volume}
  {7}},\ \bibinfo {pages} {975} (\bibinfo {year} {1966})},\ \Eprint
  {http://arxiv.org/abs/http://dx.doi.org/10.1063/1.1705011}
  {http://dx.doi.org/10.1063/1.1705011} \BibitemShut {NoStop}%
\bibitem [{\citenamefont {Husem{\"o}ller}(1994)}]{husemoller1994fibre}%
  \BibitemOpen
  \bibfield  {author} {\bibinfo {author} {\bibfnamefont {D.}~\bibnamefont
  {Husem{\"o}ller}},\ }\href {https://books.google.se/books?id=DPr\_BSH89cAC}
  {\emph {\bibinfo {title} {Fibre Bundles}}},\ Graduate Texts in Mathematics\
  (\bibinfo  {publisher} {Springer},\ \bibinfo {year} {1994})\BibitemShut
  {NoStop}%
\bibitem [{\citenamefont {Lee}(2003)}]{lee2003introduction}%
  \BibitemOpen
  \bibfield  {author} {\bibinfo {author} {\bibfnamefont {J.}~\bibnamefont
  {Lee}},\ }\href {https://books.google.se/books?id=eqfgZtjQceYC} {\emph
  {\bibinfo {title} {Introduction to Smooth Manifolds}}},\ Graduate Texts in
  Mathematics\ (\bibinfo  {publisher} {Springer},\ \bibinfo {year}
  {2003})\BibitemShut {NoStop}%
\bibitem [{\citenamefont {Chinea}(1988)}]{Chinea:1988}%
  \BibitemOpen
  \bibfield  {author} {\bibinfo {author} {\bibfnamefont {F.~J.}\ \bibnamefont
  {Chinea}},\ }\href {http://stacks.iop.org/0264-9381/5/i=1/a=018} {\bibfield
  {journal} {\bibinfo  {journal} {Classical and Quantum Gravity}\ }\textbf
  {\bibinfo {volume} {5}},\ \bibinfo {pages} {135} (\bibinfo {year}
  {1988})}\BibitemShut {NoStop}%
\bibitem [{\citenamefont {Choquet-Bruhat}\ \emph {et~al.}(1982)\citenamefont
  {Choquet-Bruhat}, \citenamefont {DeWitt-Morette},\ and\ \citenamefont
  {Dillard-Bleick}}]{choquet1982analysis}%
  \BibitemOpen
  \bibfield  {author} {\bibinfo {author} {\bibfnamefont {Y.}~\bibnamefont
  {Choquet-Bruhat}}, \bibinfo {author} {\bibfnamefont {C.}~\bibnamefont
  {DeWitt-Morette}}, \ and\ \bibinfo {author} {\bibfnamefont {M.}~\bibnamefont
  {Dillard-Bleick}},\ }\href {https://books.google.se/books?id=hUWEXphqLo8C}
  {\emph {\bibinfo {title} {Analysis, Manifolds, and Physics}}},\ \bibinfo
  {number} {pt. 1}\ (\bibinfo  {publisher} {North-Holland Publishing Company},\
  \bibinfo {year} {1982})\BibitemShut {NoStop}%
\bibitem [{\citenamefont {Hall}(2003)}]{hall2003lie}%
  \BibitemOpen
  \bibfield  {author} {\bibinfo {author} {\bibfnamefont {B.}~\bibnamefont
  {Hall}},\ }\href {https://books.google.se/books?id=m1VQi8HmEwcC} {\emph
  {\bibinfo {title} {Lie Groups, Lie Algebras, and Representations: An
  Elementary Introduction}}},\ Graduate Texts in Mathematics\ (\bibinfo
  {publisher} {Springer},\ \bibinfo {year} {2003})\BibitemShut {NoStop}%
\bibitem [{\citenamefont {Conlon}(2001)}]{Conlon:2001diff}%
  \BibitemOpen
  \bibfield  {author} {\bibinfo {author} {\bibfnamefont {L.}~\bibnamefont
  {Conlon}},\ }\href@noop {} {\emph {\bibinfo {title} {Differentiable
  Manifolds}}},\ Modern Birkh{\"a}user Classics\ (\bibinfo  {publisher}
  {Birkh{\"a}user Boston},\ \bibinfo {year} {2001})\BibitemShut {NoStop}%
\bibitem [{\citenamefont {Hall}\ and\ \citenamefont
  {McIntosh}(1983)}]{Hall:1983a}%
  \BibitemOpen
  \bibfield  {author} {\bibinfo {author} {\bibfnamefont {G.~S.}\ \bibnamefont
  {Hall}}\ and\ \bibinfo {author} {\bibfnamefont {C.~B.~G.}\ \bibnamefont
  {McIntosh}},\ }\href {\doibase 10.1007/BF02083290} {\bibfield  {journal}
  {\bibinfo  {journal} {International Journal of Theoretical Physics}\ }\textbf
  {\bibinfo {volume} {22}},\ \bibinfo {pages} {469} (\bibinfo {year}
  {1983})}\BibitemShut {NoStop}%
\bibitem [{\citenamefont {Hall}(1983)}]{Hall:1983b}%
  \BibitemOpen
  \bibfield  {author} {\bibinfo {author} {\bibfnamefont {G.~S.}\ \bibnamefont
  {Hall}},\ }\href {\doibase 10.1007/BF00759572} {\bibfield  {journal}
  {\bibinfo  {journal} {General Relativity and Gravitation}\ }\textbf {\bibinfo
  {volume} {15}},\ \bibinfo {pages} {581} (\bibinfo {year} {1983})}\BibitemShut
  {NoStop}%
\bibitem [{\citenamefont {Stephani}\ \emph {et~al.}(2003)\citenamefont
  {Stephani}, \citenamefont {Kramer}, \citenamefont {MacCallum}, \citenamefont
  {Hoenselaers},\ and\ \citenamefont
  {Herlt}}]{stephani_kramer_maccallum_hoenselaers_herlt_2003}%
  \BibitemOpen
  \bibfield  {author} {\bibinfo {author} {\bibfnamefont {H.}~\bibnamefont
  {Stephani}}, \bibinfo {author} {\bibfnamefont {D.}~\bibnamefont {Kramer}},
  \bibinfo {author} {\bibfnamefont {M.}~\bibnamefont {MacCallum}}, \bibinfo
  {author} {\bibfnamefont {C.}~\bibnamefont {Hoenselaers}}, \ and\ \bibinfo
  {author} {\bibfnamefont {E.}~\bibnamefont {Herlt}},\ }\href {\doibase
  10.1017/CBO9780511535185} {\emph {\bibinfo {title} {Exact Solutions of
  Einstein's Field Equations}}},\ \bibinfo {edition} {2nd}\ ed.,\ Cambridge
  Monographs on Mathematical Physics\ (\bibinfo  {publisher} {Cambridge
  University Press},\ \bibinfo {year} {2003})\BibitemShut {NoStop}%
\bibitem [{\citenamefont {Baccetti}\ \emph {et~al.}(2012)\citenamefont
  {Baccetti}, \citenamefont {Martin-Moruno},\ and\ \citenamefont
  {Visser}}]{Baccetti:2012re}%
  \BibitemOpen
  \bibfield  {author} {\bibinfo {author} {\bibfnamefont {V.}~\bibnamefont
  {Baccetti}}, \bibinfo {author} {\bibfnamefont {P.}~\bibnamefont
  {Martin-Moruno}}, \ and\ \bibinfo {author} {\bibfnamefont {M.}~\bibnamefont
  {Visser}},\ }\href {\doibase 10.1007/JHEP08(2012)148} {\bibfield  {journal}
  {\bibinfo  {journal} {JHEP}\ }\textbf {\bibinfo {volume} {08}},\ \bibinfo
  {pages} {148} (\bibinfo {year} {2012})},\ \Eprint
  {http://arxiv.org/abs/1206.3814} {arXiv:1206.3814 [gr-qc]} \BibitemShut
  {NoStop}%
\end{thebibliography}

\end{document}